\documentclass[conference]{IEEEtran}
\IEEEoverridecommandlockouts

\usepackage{cite}
\usepackage{multirow}
\usepackage{physics}
\usepackage{algorithm}
\usepackage{float}
\usepackage{algpseudocode}
\usepackage{amsmath}
\usepackage{amssymb}
\usepackage{amsthm}
\usepackage{graphicx}
\usepackage[labelformat=simple]{subcaption}
\usepackage{xcolor}
\usepackage{booktabs}
\usepackage{tablefootnote}
\usepackage{array}
\usepackage{tabularx}
\usepackage{environ}
\usepackage{bbding}
\usepackage{tikz}
\usepackage{pgfplots}
\usepackage{xurl}
\usepackage[hidelinks]{hyperref}
\usepackage{qcircuit}
\usepackage[normalem]{ulem}
\pgfplotsset{compat=1.18}
\usetikzlibrary{positioning,arrows.meta,patterns}
\usetikzlibrary{shapes.geometric,matrix,datavisualization,shapes.arrows,calc,backgrounds}

\newcommand{\revisioncolor}[1]{}
\newcommand{\revisionfillred}[1]{white}

\newif\ifmarkrev \markrevfalse
\ifmarkrev
  \newcommand{\rB}[1]{{\color{blue}#1}}       
  \newcommand{\rG}[1]{{\color{green!55!black}#1}} 
  \newcommand{\rR}[1]{{\color{red}#1}}         
  \newcommand{\rY}[1]{{\color{orange}#1}}      
  \colorlet{rkR}{red}\colorlet{rkB}{blue} 
\else
  \newcommand{\rB}[1]{#1}\newcommand{\rG}[1]{#1}\newcommand{\rR}[1]{#1}\newcommand{\rY}[1]{#1}
  \colorlet{rkR}{black}\colorlet{rkB}{black}
\fi

\newif\ifshowcodeurl
\showcodeurltrue
\newif\ifincludeappendix
\includeappendixtrue 

\title{Quantum Computing and Data Processing for Frequent Itemset Mining}

\author{%
\IEEEauthorblockN{Yen-Hsin Hsu\IEEEauthorrefmark{1},
Ya-Wen Teng\IEEEauthorrefmark{2},
De-Nian Yang\IEEEauthorrefmark{3},
Wang-Chien Lee\IEEEauthorrefmark{4},
Philip S. Yu\IEEEauthorrefmark{5}, and
Ming-Syan Chen\IEEEauthorrefmark{1}}
\IEEEauthorblockA{\IEEEauthorrefmark{1}National Taiwan University, Taiwan; b05203047@csie.ntu.edu.tw, mschen@ntu.edu.tw}
\IEEEauthorblockA{\IEEEauthorrefmark{2}National Chengchi University, Taiwan; ywteng@nccu.edu.tw}
\IEEEauthorblockA{\IEEEauthorrefmark{3}Academia Sinica, Taiwan; dnyang@iis.sinica.edu.tw}
\IEEEauthorblockA{\IEEEauthorrefmark{4}The Pennsylvania State University, USA; wul2@psu.edu}
\IEEEauthorblockA{\IEEEauthorrefmark{5}University of Illinois Chicago, USA; psyu@uic.edu}
}

\makeatletter
\newcommand{\problemtitle}[1]{\gdef\@problemtitle{#1}}
\newcommand{\probleminput}[1]{\gdef\@probleminput{#1}}
\newcommand{\problemquestion}[1]{\gdef\@problemquestion{#1}}

\makeatother

\NewEnviron{problem}[1]{%
\begin{center}\fbox{\parbox{3in}{%
    {\centering\scshape #1\par}%
    \parskip=1ex
    \everypar{\hangindent=1em}%
    \BODY
}}\end{center}}

\newtheorem{theorem}{Theorem}[section]
\theoremstyle{definition}
\newtheorem{definition}{Definition}
\newtheorem{corollary}{Corollary}[section]
\newtheorem{lemma}[theorem]{Lemma}
\newtheorem{example}{Example}
\newcommand\xqed[1]{\leavevmode\unskip\penalty9999 \hbox{}\nobreak\hfill\quad\hbox{#1}}
\newcommand\exend{\xqed{$\blacksquare$}}

\newcommand{\fullalgo}{Quantum Frequent-itemset Mining}
\newcommand{\onlineref}[1]{\ifincludeappendix Appendix~\ref{#1}\else\cite{online}\fi}

\begin{document}

\maketitle

\begin{abstract}
Frequent Itemset Mining (FIM) is an important task in data analytics, where classical algorithms \rB{face} scalability bottlenecks \rB{from} the combinatorial growth of candidates and \rB{the} memory overhead \rB{of their data structures}. Inspired by recent development\rB{s} in quantum computing, in this paper, we propose the \textit{\fullalgo\ (QFM)} data-processing framework for FIM. Following \rB{the} level-wise structure of the itemset lattice, QFM introduces three mechanisms: (1)~\textit{Bit-Vector Qubit Encoding} for quantum data representation, which organizes transaction data into branchless bit-vectors to facilitate systematic uncomputation; (2)~\textit{Mining-Aware Candidate Superposition}, which prepares a quantum superposition over valid candidates \rB{at each lattice level} rather than the full itemset lattice; and (3)~\textit{Bit-Parallel Threshold Marking}, which constructs a logarithmic-depth threshold-marking oracle for reliable repeated support verification within hardware coherence limits. We provide theoretical \rB{time complexity analysis}, implement QFM on IBM Qiskit and Amazon Braket, and evaluate it on real-world datasets against representative classical baselines, where QFM achieves 96\% improvement on average.

\end{abstract}


\section{Introduction}
\label{sec:intro}

Quantum computing has transitioned from early theoretical conjectures to a rapidly advancing physical reality. 
While Feynman~\cite{feynman1982simulating} and Deutsch~\cite{deutsch1985quantum} laid the groundwork for quantum computing decades ago, many people remain skeptical. Nevertheless, recent breakthroughs in hardware engineering~\cite{ibmroadmap2024,wright2019ionq,madsen2022borealis}, e.g., Google's 2019 quantum supremacy demonstration with its 53-qubit \textit{Sycamore} processor~\cite{arute2019quantum}, have shifted the focus from theoretical debate to real-world implementation.
Many industrial advances further \rB{suggest steady progress}. In 2024, IBM achieved hardware throughputs exceeding hundreds of thousands of operations per second~\cite{ibm_quantum_hardware}, while Microsoft \rB{reported progress toward} reliable, error-corrected qubits~\cite{microsoft_atom_computing_2024}. Furthermore, Google's late-2024 \textit{Willow} chip \rB{reports that a scaled-up chip can} complete complex tasks in minutes that would cost classical supercomputers septillions of years~\cite{google_willow_quantum_chip_2024}. 
The extraordinary potential of this paradigm shift stems from quantum \textit{superposition}, which allows an $n$-qubit register to simultaneously represent a state space of $2^n$ possibilities.\footnote{A qubit (quantum bit) is the fundamental unit of quantum information. Unlike a classical bit that stores a deterministic value of exactly 0 or 1, a qubit can maintain a probability distribution over both 0 and 1 simultaneously until it is explicitly measured.}
By amplifying the probability of desired outcomes~\cite{biamonte2017quantum}, quantum algorithms can effectively bypass exhaustive sequential search to achieve massive parallel processing. This property provides a critical opportunity to address intrinsic computational bottlenecks. Some foundational quantum algorithms have already demonstrated immense potential, e.g., Shor's algorithm achieves exponential speedup for factorization~\cite{shor1994algorithms}, while Grover's algorithm and quantum counting provide quadratic speedups for unstructured quantum search~\cite{grover1996fast} and frequency estimation~\cite{brassard2002quantum}, respectively. It is envisaged that these quantum capabilities can effectively improve the efficiency of traditional database and data mining tasks, such as Frequent Itemset Mining (FIM). Since FIM inherently involves exploring an exponentially growing candidate lattice and computing support counts across massive transactions, the intrinsic advantages of quantum parallelism via superposition offer a promising solution to break through classical performance limits.

The landscape of FIM has continued to evolve from classic itemset mining algorithms, such as Apriori, FP-Growth, and Eclat~\cite{Agrawal94,han2004fptree,zaki2000scalable}, toward high-utility and weighted mining~\cite{hamm2023tkde,zida2016efim,liu2012fhm}, top-$k$ pattern discovery~\cite{iqbal2021tkfim,wan2024ptf}, and stream, incremental, and constraint-based mining~\cite{dong2023fpstreampp,duong2021gctree,nguyen2020satfim}. Researchers have also achieved performance gains via distributed frameworks and GPU acceleration~\cite{singh2019sparkapriori,diaby2017sfpg,wan2024ptf}. 
Despite extensive algorithmic and architectural optimizations, existing methods are designed for classical execution models, in which candidates must be stored as explicit memory objects and support counts are maintained via sequential database scans or complex tree structures. Consequently, they inevitably suffer from severe memory bottlenecks and excessive computational overhead~\rB{\cite{borgelt2023survey,moens2013frequent}} when dealing with massive databases or low minimum support thresholds. Today, quantum computing presents a unique opportunity to fundamentally bypass these limitations. Unlike classical computing paradigms that must explicitly instantiate every candidate as a separate memory object and iteratively scan transactions to verify support, a quantum-based approach may leverage \textit{superposition} to maintain a vast candidate space within a compact quantum register. By evaluating the entire candidate space collectively and simultaneously rather than sequentially, quantum computing eliminates the need for exhaustive database scans and memory allocation, thereby offering a powerful paradigm to overcome the scalability bottlenecks of classical mining algorithms.

While these quantum advantages are highly promising, directly porting classical FIM algorithms to a quantum framework is impractical. Quantum computing operates under unique physical constraints, such as strict computational reversibility~\rB{\cite{bennett1973logical}} and the high cost of maintaining quantum coherence, which introduce entirely new design challenges.\footnote{\textit{Strict reversibility} requires quantum logic gates to be invertible. Unlike classical variables that can be freely overwritten or garbage-collected, qubits used for temporary storage cannot simply be erased after they have served their purpose. Once their intermediate values are no longer needed, and before these qubits can be reused without carrying residual garbage, they must be recycled by explicitly ``undoing'' the calculation, i.e., applying the inverse operations to reset them back to zero (a process called \textit{uncomputing}). Furthermore, unlike classical memory (e.g., RAM) that stably retains its state indefinitely, qubits possess a fragile lifespan known as \textit{quantum coherence}. If a computational circuit is too deep (i.e., execution takes too long), the qubits will degrade into random noise (decoherence) before the algorithm completes.}
Namely, a reliable quantum formulation for FIM needs to address three intertwined challenges.
\textbf{C1)~Reversible data-access mismatch.}
Classical FIM algorithms achieve highly efficient support counting by employing compressed tree- or prefix-based structures such as FP-tree~\cite{han2004fptree}, H-struct~\cite{pei2001hmine}, and related variants (e.g., CP-Tree~\cite{tanbeer2008cp} and CanTree~\cite{leung2005cantree}). These structures are designed for sequential traversal, where the next node\rB{,} chosen by following a data-dependent pointer\rB{,} immediately overwrites the temporary memory.
However, quantum computing requires all operations to be strictly reversible in order to prevent memory overwriting from \rB{breaking the required reversibility}.
If data-dependent pointer traversals are applied, the quantum system cannot simply overwrite memory states. Instead, it must allocate ancillary qubits\footnote{Unlike classical CPUs that freely overwrite temporary variables, quantum computers must allocate extra working memory (ancillary qubits) to log the history of every data-dependent jump. This record is physically mandatory to successfully reverse the circuit and clear the memory later.} to record the exact branching history for later uncomputation, which rapidly exhausts memory with unmanageable garbage qubits. It is therefore crucial to design a structurally predictable data representation that guarantees deterministic uncomputation without sacrificing the high-throughput counting efficiency of classical data structures.
\textbf{C2)~Candidate-domain encoding gap.}
The search space for itemset mining (i.e., the set of all subsets of the item universe) is exponentially larger than the set of valid frequent patterns~\cite{Agrawal94,han2000fp}. Classical FIM algorithms straightforwardly shrink this vast search domain by discarding invalid candidates from memory and explicitly iterating over the surviving subset. However, quantum computing encounters a conflicting challenge where leaving the space unpruned is computationally prohibitive, but directly mimicking classical pruning is fundamentally impractical. First, an unpruned exponential space poses a far more severe bottleneck for quantum algorithms. A direct quantum search over the full itemset lattice would allocate quantum probability to millions of irrelevant states, severely diluting the chance of observing a valid pattern.\footnote{Unlike classical search, where unpruned spaces only add sequential cost, quantum search is probabilistic. It distributes a total probability of $1.0$ across all states in the superposition. Millions of irrelevant states make valid patterns nearly unmeasurable, requiring more amplification iterations, deeper circuits, and greater risk of decoherence-induced failure.} Second, despite the urgent need to shrink the search space, directly mimicking classical approaches introduces a severe state-preparation overhead. While initializing a quantum superposition over a complete, uniform exponential space (\rB{$2^{N}$}) is straightforward (via parallel Hadamard gates~\rB{\cite{nielsen2010quantum}}), a quantum register cannot simply ``load'' a discrete list of candidate itemsets. Forcing the system to strictly encode this sparse, arbitrarily distributed list requires synthesizing deep, hardcoded circuits burdened by massive conditional logic (i.e., highly expensive multi-controlled gates) to explicitly suppress all invalid patterns. If this irregular encoding is naively performed at each mining level, the resulting circuit depth rapidly grows to an unmanageable degree. 
Therefore, it is imperative to develop a state-preparation scheme that restricts the quantum superposition to valid candidates, avoiding the deep, hardcoded circuits otherwise required to encode a massive list of specific candidate itemsets.
\textbf{C3)~Repeated reversible support-verification cost.}
Classical FIM verifies candidate support by sequentially scanning transactions or iteratively traversing trees to accumulate frequency counts. However, porting this sequential logic directly to a quantum setting threatens system stability. Because the quantum search procedure invokes a verification oracle (i.e., the subroutine computing candidate supports) repeatedly, mimicking classical counting forces the system to compile these sequential counting loops into an excessively long sequence of reversible logic gates.\footnote{Unlike classical CPUs that dynamically execute jump/loop instructions, quantum circuits generally lack dynamic control flow for data-dependent loops. A classical loop over \rB{the $M$ transactions} must be physicalized as an ``unrolled'' sequence of \rB{$M$} reversible gate blocks, which linearly deepens the circuit.} Such deep circuits demand an execution duration that severely exceeds the strict coherence time limits of physical hardware, rendering the computation fundamentally infeasible. 
Therefore, it is essential to bypass sequential accumulation by mapping global database verification onto a quantum circuit of mathematically bounded depth, ensuring reliable execution within hardware coherence limits.

In this paper, we explore quantum computing for the frequent itemset mining (FIM) problem~\cite{Agrawal94,han2000fp}. The proposed framework, \textit{\fullalgo\ (QFM)}, introduces three important components designed to address the three challenges.
1)~For \textbf{C1}, we introduce \textit{Bit-Vector Qubit Encoding} by structuring transaction data into branchless bit-vectors, to avoid the garbage-qubit overhead of data-dependent pointer traversals and strictly guarantee systematic uncomputation. Furthermore, we design a \textit{counting oracle} $\mathcal{O}_{\sf count}$ to efficiently reduce support-evaluation complexity to $\mathcal{O}(\sqrt{M})$ (where $M$ is the number of transactions), yielding a theoretical quadratic speedup over the $\mathcal{O}(M)$ bound of classical sequential scans.
2)~For \textbf{C2}, we design \textit{Mining-Aware Candidate Superposition} by mapping the sparse candidate domain onto a compact space via \textit{Candidate Subspace Index (CSI)}, strictly confining the quantum superposition to valid parent combinations rather than the full $O(2^N)$ lattice. This compact CSI representation enables efficient state preparation via parallel Hadamard gates, preventing the exponential probability dilution of full-lattice searches while completely bypassing the deep, multi-controlled logic required to encode arbitrarily distributed candidates.
3)~For \textbf{C3}, we develop \textit{Bit-Parallel Threshold Marking} by formulating a \textit{threshold-marking oracle} $\mathcal{O}_{\ge\sigma}$, which is a parallelized reversible arithmetic block, completely bypassing the unmanageable $\mathcal{O}(M)$ \rB{time} of classical sequential scans. This shift to parallel bitwise computation enforces a strictly bounded $\mathcal{O}(\log M)$ \rB{time}, ensuring the iteratively queried oracle \rB{executes} reliably within hardware coherence limits during amplitude amplification.

We theoretically analyze QFM in terms of time complexity, implement it on IBM Qiskit and Amazon Braket, and evaluate our implementation on real-world datasets against various classical baselines, including a GPU-parallel approach. 
Our evaluation confirms that QFM successfully maintains a strictly shallow circuit depth, demonstrating the practical potential of quantum parallelism to eventually bypass the persistent bottlenecks of classical mining algorithms. The contributions include:
\begin{itemize}
   \item We identify three critical quantum-specific challenges for frequent itemset mining: the reversible data-access mismatch, candidate-domain encoding gap, and repeated reversible support-verification cost. 
   \item We design a novel framework, \fullalgo\ (QFM), which involves Bit-Vector Qubit Encoding, Mining-Aware Candidate Superposition, and Bit-Parallel Threshold Marking to address the challenges.
    \item We conduct a comprehensive theoretical analysis of QFM and validate our proposal on IBM Qiskit and Amazon Braket. Extensive evaluation on real-world datasets demonstrates that QFM outperforms representative baselines by an average of 96\%.
\end{itemize}

\section{Related Work}
\label{sec:related_work}

\textbf{Classical candidate, vertical, and tree-based mining.}
Frequent itemset mining (FIM) has been well studied through several design changes. 
Apriori~\cite{Agrawal94} leverages downward-closure pruning for level-wise candidate reduction, Eclat~\cite{zaki2000scalable} evaluates support via vertical tid-list intersections, and FP-growth~\cite{han2004fptree} bypasses explicit candidate generation entirely by recursively mining prefix-tree compressed databases.\footnote{While directly adapting these paradigms is physically prohibitive, we discuss possible extensions of QFM toward quantum counterparts of these classical algorithms in~\onlineref{app:eclat_mapping}.}
To accelerate these pioneering algorithms, subsequent efforts have been made to heavily optimize them by developing advanced data representations (e.g., compact array-based variants), hashing mechanisms, and bitwise set operations~\rB{\cite{borgelt2023survey,vo2012dbv}}.
Beyond standard frequent itemsets, researchers have broadly extended these classical principles to alternative pattern definitions, including closed itemsets~\cite{pasquier1999discovering,zaki2002charm,lucchese2005fast,uno2004lcm}, maximal itemsets~\cite{chen2002fpmax,zhang2021multi,zhang2021right}, high-utility itemsets~\cite{wu2022ubpminer}, and top-$k$ mining~\cite{rehman2022efficient,iqbal2021tkfim}.
However, despite diverse algorithmic advancements and extended pattern definitions, their efficiency still relies on assumptions that cannot directly transfer to quantum computing, i.e., their mechanisms provide no reversible counterpart for pointer-based data access, no method for preparing quantum states restricted to valid candidate itemsets, and no low-depth support-verification structure suitable for repeated verification oracle invocations, and therefore cannot address \textbf{C1}, \textbf{C2}, and \textbf{C3} in FIM.

\textbf{Distributed and GPU-accelerated mining.}
To handle massive datasets, distributed frameworks such as Dist-Eclat~\cite{moens2013frequent}, PFP~\cite{li2008pfp}, and Spark-based EAFIM~\cite{raj2020eafim} partition mining tasks across clusters, while GPU-oriented approaches further accelerate parallel counting or itemset operations on many-core devices~\cite{diaby2017sfpg,fang2024gpu}. These studies improve the throughput of classical mining by assigning candidate generation, data partitioning, and support evaluation to multiple workers or threads. 
However, they accelerate candidate generation and support evaluation on classical hardware, instead of transferring them into quantum-compatible operations. Therefore, they also cannot reformulate data access as reversible operations, restrict quantum state preparation to valid candidate itemsets, or express repeated support verification as shallow circuits, and thus cannot address \textbf{C1}, \textbf{C2}, and \textbf{C3} in FIM.

\textbf{Quantum algorithms and quantum mining.}
Quantum computing has been studied across a wide range of problems, from foundational models of quantum query algorithms and computational complexity~\cite{ambainis2004quantum,buhrman2002complexity,brassard2002quantum}, to quantum simulation and combinatorial optimization~\cite{kandala2017vqe,farhi2014qaoa}, and to domain applications such as finance~\cite{FinanceReview2018}. These studies demonstrate the broad applicability of quantum computing, but they are not designed to address FIM. Among mining-related studies, qARM~\cite{yu2016quantum,yu2022experimental} is an early attempt at association-rule mining that partially uses quantum counting. However, it relies on a classical join-and-prune procedure in traditional machines, instead of quantum circuits for candidate generation, lacking a complete quantum mining pipeline.\footnote{qARM falls short of QFM for \rB{three} reasons. 1) It mines within a given candidate set: qARM bounds the database-oracle queries for mining frequent \(k\)-itemsets inside a \emph{given} candidate set, summed over levels, but keeps inter-level candidate generation classical. 2) It relies on approximate supports, which weakens exact thresholding: qARM outputs \(\epsilon\)-approximate supports, making thresholding nondeterministic near the boundary and introducing a \(1/\epsilon\) factor; \rB{for normalized supports, in the worst case near the threshold,} exact threshold decisions force \(\epsilon=\mathcal{O}(1/M)\). Accordingly, qARM requires \(\mathcal{O}((k{+}1)\sqrt{|C_{k+1}|\,|L_{k+1}|}/\epsilon)\), whereas Lemma~\ref{lem:qil_complexity} in Section~\ref{sec:complexity} gives QFM the exact-thresholding bound \(\mathcal{O}(\sqrt{|C_{k+1}|\max\{1,|L_{k+1}|\}}\cdot\mathrm{polylog})\). 3) Its experiments remain limited in scale: qARM compiles circuits only for \(2{\times}2\) and \(4{\times}4\) toy databases~\cite{yu2022experimental}, whereas QFM is evaluated on realistic benchmark datasets with hardware validation on Amazon Braket. A side-by-side comparison is provided in~\onlineref{sec:comparison}.}
In contrast, we formulate FIM as a quantum mining problem and identify challenges, including reversible data-access mismatch, candidate-domain encoding gap, and repeated reversible support-verification cost. We further design QFM as an end-to-end quantum FIM algorithm, provide theoretical analysis, and validate the framework on IBM Qiskit and Amazon Braket.

\section{Problem Formulation and Preliminary}
\label{sec:problem_formulation}

In this section, we first define the Frequent Itemset Mining (FIM) problem and then introduce the quantum preliminaries. The notation table is presented in~\onlineref{sec:table} due to space constraints.

\subsection{Problem Definition}

Following FIM in~\cite{Agrawal94}, let $\mathcal{I} = \{i_1, i_2, \dots, i_N\}$ be a universe of $N$ distinct items. A transaction database $D = \{T_1, T_2, \dots, T_M\}$ consists of $M$ transactions, where transaction $T_j$, denoted by transaction identifier $j$, contains a subset of items in $\mathcal{I}$ for $j=1,\ldots,M$. An itemset $I \subseteq \mathcal{I}$ is a non-empty set of items. The length of an itemset $I$ is denoted by $|I|$, and an itemset of length $k$ is referred to as a $k$-itemset.

\begin{definition}[Support~\cite{Agrawal94}]
The support count of an itemset $I$ in database $D$, denoted as $\operatorname{sup}(I)$, is the number of transactions in $D$ that contain $I$, i.e., $\operatorname{sup}(I) = |\{T_j \in D \mid I \subseteq T_j,\ \text{where } 1\leq j\leq M\}|$.
\end{definition}

Given a user-specified minimum support threshold $\sigma$ (where $1 \le \sigma \le M$), an itemset $I$ is considered \textit{frequent} if $\operatorname{sup}(I) \ge \sigma$. Let $L_k$ denote the set of all frequent $k$-itemsets, and $L = \bigcup_{k} L_k$ be the set of all frequent itemsets. 

\begin{definition}[Frequent Itemset Mining (FIM)]
Given a transaction database $D$ and a minimum support threshold $\sigma$, the FIM problem aims to identify the complete set of frequent itemsets $L = \{I \subseteq \mathcal{I} \mid \operatorname{sup}(I) \ge \sigma\}$.
\end{definition}

\subsection{Preliminary of Quantum Computing}
\label{sec:preliminary} 
In this section, we first describe quantum states, including qubits, basis states, measurement, and superposition, to explain how quantum information is represented and read. We then introduce quantum circuits, including gates, circuit depth, reversibility, and uncomputation, to explain how quantum computing is organized and why temporary values need to be cleaned carefully. Finally, we describe oracles and marking, which are used to express condition checking inside quantum computing.

\paragraph{Quantum states}
A \textit{qubit} is the basic unit of quantum information. While a classical bit has a deterministic value of either $0$ or $1$, a qubit can be described by probabilities over the two outcomes $0$ and $1$ before it is measured. These two measurable outcomes are called the \textit{basis states}, denoted by $\ket{0}$ and $\ket{1}$.\footnote{We use Dirac notation $\ket{x}$ to denote computational basis states.} Measurement reads the qubit and returns one of these basis states according to its probabilities.
Multiple qubits can be grouped as a \textit{register} to encode a binary label. For example, a register with $d$ qubits has $2^d$ basis states $\ket{0},\ldots,\ket{2^d-1}$, each of which can serve as one label in a finite search domain. \textit{Superposition} allows the register, before measurement, to keep multiple labels possible at the same time. A state with the form $\sum_{r=0}^{2^d-1}\alpha_r\ket{r}$ represents this superposition over labels, where the coefficients $\alpha_r$ determine the probability of each label being measured.

\paragraph{Quantum circuits}
Quantum computing is performed on a \textit{circuit}, which is a sequence of \textit{gates} that change the states of qubits or registers. The \textit{circuit depth} counts the number of dependent gate layers in a circuit. Gates in the same layer can be executed in parallel, while gates with data dependency need to be placed in later layers. A shallower circuit is generally preferable because it reduces the amount of sequential quantum computation~\rB{\cite{preskill2018nisq}}.
A fundamental property of quantum circuits is \textit{reversibility}~\rB{\cite{bennett1973logical,nielsen2010quantum}}. Each gate requires an invertible transformation, meaning the previous state can be recovered from the output state. As a result, operations that are routine in classical programs, such as overwriting or discarding temporary values, must be represented through reversible transformations in quantum circuits. A quantum circuit cannot simply erase intermediate values, because this would remove information needed to reverse the computation, making the operation invalid as a quantum circuit step. This imposes a strict constraint on quantum circuit design, since every intermediate value must be carefully managed and either preserved as part of a reversible mapping or explicitly uncomputed before the workspace can be safely reused. The standard technique is \textit{uncomputation}~\rB{\cite{bennett1973logical}}, which cleans temporary values by applying inverse operations to return the corresponding registers to their original states, so that these registers can be reused in later circuit steps without carrying irrelevant intermediate information.

\paragraph{Oracles and marking}
An \textit{oracle} is a reusable \rB{and} reversible circuit that checks whether a state satisfies a specified condition, such as whether a candidate itemset reaches the minimum-support threshold. Since directly measuring a superposition would collapse it, the oracle does not enumerate desired states by measurement. Instead, it \textit{marks} the states satisfying the condition, so that later amplification steps~\rB{\cite{grover1996fast,brassard2002quantum}} can make these states very likely to appear as the desired observed results.

\section{Algorithm QFM}
\label{sec:algorithm}

In this section, we propose \textit{\fullalgo\ (QFM)}. To address the reversible data-access mismatch (\textbf{C1}), we introduce \textit{Bit-Vector Qubit Encoding}. Rather than \rB{recursively navigating} pointer-based data structures (which forces a quantum system to retain unpredictable traversal paths), this encoding explicitly structures the transaction data as branchless bit-vectors to ensure that temporary quantum registers can be systematically uncomputed and recycled, effectively avoiding unmanageable garbage-qubit overhead~\rB{\cite{bennett1973logical}}. Furthermore, we design a \textit{counting oracle} $\mathcal{O}_{\sf count}$ to reduce the required oracle calls for \rB{per-item} support evaluation to $\mathcal{O}(\sqrt{M})$ with quantum amplitude estimation (QAE)~\cite{brassard2002quantum}. This yields a quadratic theoretical speedup over the $\mathcal{O}(M)$ sequential scans inherently required by classical algorithms~\cite{Agrawal94,han2004fptree}.

To address the candidate-domain encoding gap (\textbf{C2}), we design \textit{Mining-Aware Candidate (MAC) Superposition}. Instead of initializing a quantum state over the full $O(2^N)$ itemset lattice or synthesizing deep, hardcoded circuits to explicitly encode a sparse, arbitrarily distributed candidate list, MAC \rB{s}uperposition mathematically maps the search domain onto a contiguous register via \rB{the} \textit{Candidate Subspace Index (CSI)}. Each CSI \rB{index} references two combinable frequent itemsets discovered in the previous mining level, 
avoiding deep conditional circuits required to hardcode every single valid candidate combination into a quantum state. This enables \rB{the} MAC \rB{s}uperposition to be prepared with a highly optimized circuit of shallow depth.
This design decouples candidate-domain construction from support verification by first restricting the quantum state to \rB{the CSI-indexed candidate domain}. Thus, support \rB{verification} is \rB{then} invoked only on candidates represented by those CSI \rB{indices}. The search space is therefore limited to the number of structurally valid parent pairs, rather than the full $O(2^N)$ itemset lattice, drastically reducing circuit complexity and preventing the threshold-marking oracle from wasting expensive resources on invalid itemsets.

To mitigate the repeated reversible support-verification cost (\textbf{C3}), we develop \textit{Bit-Parallel Threshold Marking}. Rather than unrolling classical sequential scans into unmanageably deep linear circuits, we design the \textit{threshold-marking oracle} $\mathcal{O}_{\ge\sigma}$ as a highly parallelized reversible arithmetic block. By exploiting Bit-Vector Qubit Encoding, $\mathcal{O}_{\ge\sigma}$ turns support verification into a bit-vector counting problem, allowing the support count to be computed in parallel rather than accumulated through sequential scans. It reduces the oracle \rB{evaluation time} from $\mathcal{O}(M)$ to $\mathcal{O}(\log M)$. This linear-to-logarithmic reduction is critical because the threshold-marking oracle \rB{needs to be} iteratively queried during amplitude amplification.

Equipped with the above designs, QFM operates as an iterative mining pipeline with three modules: Quantum Preprocessing and Representation (QPR), Quantum Support Verification (QSV), and Quantum Itemset Listing (QIL). QFM begins with QPR, which takes the transaction database and the minimum support threshold as inputs. QPR \rB{first} employs $\mathcal{O}_{\sf count}$ to evaluate 1-item supports, identifying the frequent 1-itemsets and initializing their corresponding Bit-Vector Qubit Encoding. After the 1-itemset stage, QFM iterates between QSV and QIL to mine increasingly \rB{large} itemsets. Given the frequent itemsets discovered at the previous iteration, QSV prepares the Mining-Aware Candidate Superposition and configures the reusable $\mathcal{O}_{\ge\sigma}$ to evaluate candidate supports and mark those satisfying the support threshold. QIL then extracts the candidates that $\mathcal{O}_{\ge\sigma}$ marks as frequent. This iterative process continues until no new candidates can be formed, ultimately yielding the complete set of frequent itemsets.
\rR{Algorithm~\ref{alg:QFM} presents the main pseudocode; \rB{additional} module-level procedures are provided in~\onlineref{sec:codes}.}

\begin{algorithm}[t]
  \caption{QFM\rB{. $U_{\rm cand}^{(k+1)}$ denotes the Mining-Aware Candidate superposition prepared by QSV over the level-$(k{+}1)$ candidate domain.}}
  \label{alg:QFM}
  \begin{algorithmic}[1]
    \Require transaction database $D$, minimum support threshold $\sigma$
    \Ensure all frequent itemsets $L$
    \State $L_1 \gets QPR(D,\sigma)$ 
    \State $L\gets L_1$, $k\gets 1$
    \While{$L_k \neq \emptyset$}
        \State $(U_{\rm cand}^{(k+1)}, \mathcal O^{(k+1)}_{\ge\sigma}) \gets \rB{QSV}(L_k,\sigma)$ 
        \State $L_{k+1} \gets QIL(U_{\rm cand}^{(k+1)},\mathcal O^{(k+1)}_{\ge\sigma})$
        \If{$L_{k+1}=\emptyset$}
            \State \textbf{break}
        \EndIf
        \State $L\gets L\cup L_{k+1}$; $k\gets k+1$
    \EndWhile
    \State \Return $L$
  \end{algorithmic}
\end{algorithm}

\subsection{Quantum Preprocessing and Representation (QPR)}
\label{subsec:QPr}
\begin{figure}[t]
\centering
{\color{rkR}\small
\setlength{\tabcolsep}{7pt}\renewcommand{\arraystretch}{1.18}
\begin{tabular}{l|cccc}
TID & $T_1$ & $T_2$ & $T_3$ & $T_4$ \\ \hline
Items & A, C, D & B, C, E & A, B, C, E & B, E
\end{tabular}}
\caption{\rR{Running-example transaction database ($\sigma=2$).}}
\label{fig:toy_db}
\end{figure}

The primary objective of QPR is to initialize the mining session by replacing recursively created data structures in classical algorithms with a deterministic, branchless data representation that works seamlessly with subsequent quantum oracles. Instead of navigating unpredictable, pointer-based structures (e.g., FP-tree), QPR performs a one-time systematic ingestion of the raw transaction database $D$ and maps $D$ into a uniform bit-vector space. Given the threshold $\sigma$, QPR identifies all frequent 1-itemsets (denoted as $L_1$) and constructs their \textit{Bit-Vector Qubit Encoding}. Ultimately, this initialization decouples the data representation from arbitrary traversal paths, establishing a structurally predictable foundation that strictly guarantees systematic uncomputation for all later mining stages.

To construct Bit-Vector Qubit Encoding, we design the \textit{counting oracle $\mathcal{O}_{\sf count}$}. In classical FIM algorithms, determining the support of an item $i$ involves a linear scan of the database $D$ with time complexity $\mathcal{O}(M)$. To overcome this $\mathcal{O}(M)$ bottleneck while strictly adhering to quantum reversibility, we define $\mathcal{O}_{\sf count}$ as a \textit{unitary operator} that retrieves item occurrence values.
\begin{definition}[Counting Oracle $\mathcal{O}_{\sf count}$]
For each item $i$, the counting oracle is a unitary operator that evaluates item presence. Given a transaction index $j$, it indicates whether item $i$ appears in transaction $j$ via:
\begin{equation}
  \mathcal{O}_{\sf count}:\ket{j,b}\;\mapsto\;\ket{j,b\oplus X_{i,j}},
\end{equation}
where $b\in\{0,1\}$ is the state of a target qubit, $\oplus$ denotes the XOR operation, and $X$ is a binary incidence matrix such that $X_{i,j}=1$ if and only if item $i$ appears in transaction $j$.
Since one such oracle is associated with each item, we write $\mathcal{O}^{(i)}_{\sf count}$ when the item needs to be emphasized.
\end{definition}

\begin{example}
\label{ex:count_oracle}
Consider verifying whether item A is present in transaction $T_3$ of the database shown in \rR{Figure~\ref{fig:toy_db}}. With the target qubit initialized to $b=0$, QFM begins at the state $\ket{3,0}$. Executing the counting oracle yields:
$$\mathcal{O}_{\sf count}\ket{3,0} \mapsto \ket{3,0 \oplus \rB{X_{A,3}}} = \ket{3,0 \oplus 1} = \ket{3,1}.$$
The resulting state $\ket{3,1}$ indicates that A is present in $T_3$.
\exend
\end{example}

To initialize the mining session, QPR evaluates each item $i \in \mathcal{I}$ by preparing a quantum superposition over all $M$ transactions. In this quantum state, each transaction is assigned a probability \textit{amplitude}, and the squared magnitude of the amplitude determines the probability of observing that transaction. Rather than performing a classical linear scan, QPR repeatedly interleaves the counting oracle $\mathcal{O}_{\sf count}$ (which flags the specific transactions containing item $i$) with a reflection operator.\footnote{This operator drives quantum interference, a mechanism that combines amplitudes to constructively reinforce desired outcomes and cancel out undesired ones. It geometrically reflects the amplitude values across their overall average, systematically increasing the amplitudes of the flagged transactions while suppressing the unflagged ones.} Crucially, because the actual probability scales with the square of these amplified \textit{amplitudes}, \rB{the detection probability of the flagged transactions grows quadratically with the number of amplification iterations in the initial amplification regime}. This fundamental property \rB{enables} the quantum state to converge and reveal \rB{a bounded-error} support count estimate $\operatorname{sup}(i)$ using only $\mathcal{O}(\sqrt{M})$ oracle calls, effectively breaking the classical $\mathcal{O}(M)$ scanning limit. The resulting estimates are utilized to filter out infrequent items, systematically pruning them to conserve computational resources.

For the surviving frequent 1-itemsets $L_1$, each itemset $I$ can be represented as a transaction-indicator vector $v_I \in \{0,1\}^{M}$, where $(v_I)[j]=1$ if and only if $I \subseteq T_j$; hence, its support is $\operatorname{sup}(I)=\|v_I\|_1$. Next, $v_I$ serves as the fundamental blueprint to initialize the \textit{Bit-Vector Qubit Encoding}. Instead of creating an amplitude-style quantum superposition over transaction indices, this encoding explicitly maps $v_I$ into an $M$-qubit register as a deterministic computational basis state (i.e., $\ket{v_I} = \ket{(v_I)[1]} \otimes \dots \otimes \ket{(v_I)[M]}$). Elevating the transaction-indicator vector into the quantum domain is essential to allow downstream quantum operations to interact seamlessly with candidate superpositions. Consequently, QPR identifies $L_1$ \rB{within} \rB{$\mathcal{O}(N\sqrt{M})$} oracle calls and establishes strictly predictable, branchless Bit-Vector Qubit Encoding, ensuring that subsequent circuits can process transaction data without the unmanageable garbage-qubit overhead of pointer-based traversals.

\begin{example}
\label{ex:QPr}
Consider the database presented in \rR{Figure~\ref{fig:toy_db}} under a minimum support threshold $\sigma = 2$. QPR represents each 1-itemset \rB{by} the transaction-indicator vector\rB{s} below:
\[
\begin{array}{lll}
v_{\text{A}} = [1010] & v_{\text{B}} = [0111] & v_{\text{C}} = [1110] \\
v_{\text{D}} = [1000] & v_{\text{E}} = [0111] &
\end{array}
\]
QPR then leverages $\mathcal{O}_{\sf count}$ to evaluate the support count of each item across all transactions. The derived support counts for 1-itemsets are as follows:
\[
\begin{array}{lll}
\operatorname{sup}(\text{A}) = 2 &
\operatorname{sup}(\text{B}) = 3 &
\operatorname{sup}(\text{C}) = 3 \\
\operatorname{sup}(\text{D}) = 1 &
\operatorname{sup}(\text{E}) = 3 &
\end{array}
\]
Finally, QPR constructs the Bit-Vector Qubit Encoding for the frequent 1-itemsets (i.e., A, B, C, and E) as follows.
\[
\begin{aligned}
\ket{v_\text{A}} &= \ket{1}\otimes\ket{0}\otimes\ket{1}\otimes\ket{0}
          = \ket{1010}, \\
\ket{v_\text{B}} &= \ket{0}\otimes\ket{1}\otimes\ket{1}\otimes\ket{1}
          = \ket{0111}, \\
\ket{v_\text{C}} &= \ket{1}\otimes\ket{1}\otimes\ket{1}\otimes\ket{0}
          = \ket{1110}, \\
\ket{v_\text{E}} &= \ket{0}\otimes\ket{1}\otimes\ket{1}\otimes\ket{1}
          = \ket{0111}.
\end{aligned}
\]
\exend
\end{example}

\subsection{Quantum Support Verification (QSV)}
\label{subsec:QFc}
\label{subsec:QSV} 
Following QPR, QFM progressively mines larger itemsets in level-wise loop (where each level $k$ utilizes the frequent $k$-itemsets, denoted as $L_k$, to discover candidates of size $k+1$). Within a loop, i)~QSV prepares a \textit{Mining-Aware Candidate (MAC) superposition} comprising only \rB{structurally} valid candidates by leveraging the \textit{Candidate Subspace Index (CSI)} to restrict the quantum register to prefix-sharing joins of itemsets in $L_k$; and ii)~\rB{QSV constructs} a reusable threshold-marking oracle that evaluates support counts exclusively for these candidates, requiring only logarithmic circuit depth.
Procedure~\ref{alg:QSV} presents the pseudocode of QSV.

\begingroup
\floatname{algorithm}{Procedure}
\begin{algorithm}[t]
  \caption{Quantum Support Verification (\rB{QSV})}
  \label{alg:QSV}
  \begin{algorithmic}[1]
    \Require frequent itemsets $L_k$, threshold $\sigma$
    \Ensure candidate preparation $U_{\rm cand}^{(k+1)}$ and threshold oracle $\mathcal O^{(k+1)}_{\ge\sigma}$
    \State Construct prefix-sharing parent-pair set $P_k=\{(u,w):u,w\in L_k, |u\cup w|=k+1\}$
    \State Define $U_{\rm cand}^{(k+1)}$ that prepares a superposition over $P_k$ via CSI
    \State Define load oracle $\mathcal O^{(k)}_{\sf load}: |x,z\rangle\mapsto |x,z\oplus v_x \rangle$
    \State Define $\mathcal O^{(k+1)}_{\ge\sigma}$ on $|u,w,b\rangle$:
    \Statex \quad (i) load $v_u,v_w$ using $\mathcal O^{(k)}_{\sf load}$
    \Statex \quad (ii) compute $v_{u\cup w}=v_u\wedge v_w$ and $s=\|v_{u\cup w}\|_1$
    \Statex \quad (iii) mark $b$ iff $s\ge\sigma$
    \Statex \quad (iv) uncompute all work registers
    \State \Return $U_{\rm cand}^{(k+1)},\mathcal O^{(k+1)}_{\ge\sigma}$
  \end{algorithmic}
\end{algorithm}
\endgroup

Let $P_k$ denote the set of prefix-sharing parent pairs $\{(u,w) \mid u,w\in L_k \text{ and } |u\cup w|=k+1 \}$; \rB{under a fixed item order, such a pair consists of two frequent $k$-itemsets that share their first $k{-}1$ items and differ in the last, so each candidate $(k{+}1)$-itemset corresponds to exactly one pair in $P_k$}.
QSV prepares \rB{the} MAC superposition over $P_k$ structured by CSI, where each basis state $\ket{u,w}$ encodes exactly one parent pair $(u,w)$. By encoding only these parent indices, each state implicitly defines the corresponding candidate itemset $I=u\cup w$. Specifically, QSV constructs CSI by numbering the parent pairs in $P_k$ as a contiguous index range through a bijective map
\[
\pi_k:\{0,\ldots,|P_k|-1\}\rightarrow P_k,
\qquad
\pi_k(r)=(u_r,w_r).
\]
It then allocates a \(q_k\)-qubit index register, where \(q_k=\lceil\log_2 |P_k|\rceil\) \rB{(well-defined since the level-wise iteration proceeds only when $P_k\neq\emptyset$)}, initialized to \(\ket{0}^{\otimes q_k}\). \rB{A parallel layer of Hadamard gates prepares the superposition in \(\mathcal{O}(1)\) time:}
\[
H^{\otimes q_k}\ket{0}^{\otimes q_k}
=
\frac{1}{\sqrt{2^{q_k}}}
\sum_{r=0}^{2^{q_k}-1}\ket{r}.
\]
For each valid index \(r<|P_k|\) within CSI, the basis state \(\ket{r}\) represents the parent pair \(\pi_k(r)=(u_r,w_r)\), and we write this logical state as \(\ket{u_r,w_r}\). If \(2^{q_k}>|P_k|\), the remaining basis states are padded indices with no associated parent pair, and thus never marked by the threshold-marking oracle.

\begin{example}
\label{ex:superposition}
Following Example~\ref{ex:QPr}, the frequent 1-itemsets are A, B, C, and E. For \(k=1\), QSV forms the parent-pair domain
\[
P_1=
\{(\text{A},\text{B}),(\text{A},\text{C}),(\text{A},\text{E}),
(\text{B},\text{C}),(\text{B},\text{E}),(\text{C},\text{E})\}.
\]
Each pair of 1-itemsets are unioned to give the candidate 2-itemsets
\[
C_2 =
\{\{\text{A},\text{B}\}, \{\text{A},\text{C}\}, \{\text{A},\text{E}\},
\{\text{B},\text{C}\}, \{\text{B},\text{E}\}, \{\text{C},\text{E}\}\}.
\]
The six parent pairs are numbered by CSI as follows:
\[
\begin{array}{lll}
0\leftrightarrow(\text{A},\text{B}) &
1\leftrightarrow(\text{A},\text{C}) &
2\leftrightarrow(\text{A},\text{E}) \\
3\leftrightarrow(\text{B},\text{C}) &
4\leftrightarrow(\text{B},\text{E}) &
5\leftrightarrow(\text{C},\text{E})
\end{array}
\]
Since \(|P_1|=6\), QSV uses a \(q_1=\lceil\log_2 6\rceil=3\)-qubit index register. 
Starting from \(\ket{000}\), QSV exploits a parallel Hadamard layer to prepare the Mining-Aware Candidate Superposition:
\[
H^{\otimes 3}\ket{000}
=
\frac{1}{\sqrt{8}}\sum_{r=0}^{7}\ket{r},
\]
where the valid indices \(r=0,\ldots,5\) correspond to the six parent pairs listed above, e.g., \(\ket{3}\equiv\ket{\text{B},\text{C}}\) with candidate itemset \(\{\text{B},\text{C}\}=\text{B}\cup\text{C}\), while the padded indices \(r=6,7\) have no associated parent pair.
\exend
\end{example}

With Mining-Aware Candidate Superposition established, QSV then verifies the support for each candidate $I=u \cup w$. This verification involves two oracles: a)~A \textit{transaction load oracle} $\mathcal{O}^{(k)}_{\sf load}$ to instantiate the transaction-indicator vectors of the parent itemsets as quantum states, enabling subsequent quantum operators to interact seamlessly with the candidate superposition; b)~A \textit{threshold-marking oracle} $\mathcal{O}^{(k+1)}_{\ge\sigma}$ to compute the candidate support and mark those that meet the threshold $\sigma$.

\begin{definition}[Transaction Load Oracle $\mathcal{O}^{(k)}_{\sf load}$]
Let $v_x\in\{0,1\}^M$ be the transaction-indicator vector for an itemset $x\in L_k$. The load oracle is defined as:
\begin{equation}
\mathcal O^{(k)}_{\sf load}:\ket{x,z}\mapsto \ket{x,z\oplus v_x},
\end{equation}
where $z$ is an $M$-qubit target register. Initialized to the zero state $0^M$, this register $z$ receives the transaction-indicator vector $v_x$ via XOR, resulting in the output state $\ket{x, v_x}$, i.e., mapping $v_x$ into the quantum circuit. Since XOR is self-inverse, applying the same load oracle again maps $\ket{x,z\oplus v_x}$ back to $\ket{x,z}$. Hence,
$\mathcal{O}^{(k)}_{\sf load}$ is reversible.
\end{definition}

\begin{example}
\label{ex:load_oracle}
Following Example~\ref{ex:superposition}, take the valid index \(r=3\), which corresponds to the parent pair \(\ket{3}\equiv\ket{\text{B},\text{C}}\). From Example~\ref{ex:QPr}, the corresponding parent vectors are \(v_{\text{B}}=[0111]\) and \(v_{\text{C}}=[1110]\). Since the \(M\)-qubit target register is initialized to \(\ket{0^4}\), the load oracle gives
\[
\begin{aligned}
\mathcal{O}^{(1)}_{\sf load}\ket{\text{B},0^4}
&\mapsto
\ket{\text{B},0^4\oplus v_{\text{B}}}
=
\ket{\text{B},0111}, \\
\mathcal{O}^{(1)}_{\sf load}\ket{\text{C},0^4}
&\mapsto
\ket{\text{C},0^4\oplus v_{\text{C}}}
=
\ket{\text{C},1110}.
\end{aligned}
\]
Thus, the parent vectors \(v_{\text{B}}\) and \(v_{\text{C}}\) are loaded as \(\ket{0111}\) and \(\ket{1110}\), respectively.
\exend
\end{example}

\begin{definition}[Threshold-Marking Oracle $\mathcal{O}^{(k+1)}_{\ge\sigma}$]
For a parent-pair representation of a candidate $I= u \cup w$, this oracle marks whether the candidate\rB{'s} support reaches the minimum-support threshold $\sigma$ as follows.
\begin{equation}
  \mathcal O^{(k+1)}_{\ge\sigma}:\ket{u,w,b}\mapsto \ket{u,w,b\oplus[\|v_u\wedge v_w\|_1\ge\sigma]},
\end{equation}
where $b \in \{0,1\}$ is the state of a target qubit to indicate whether $I$'s support reaches $\sigma$. The oracle is implemented by the following reversible sequence, whose circuit-level realization is shown in Figure~\ref{fig:qsv_circuit}:
\begin{equation}
\mathcal{O}^{(k+1)}_{\ge\sigma}=U_{\rm load}^{\dagger}U_{\wedge}^{\dagger}U_{\rm pop}^{\dagger}U_{\rm cmp}^{\dagger}X_{\rm flag}U_{\rm cmp}U_{\rm pop}U_{\wedge}U_{\rm load}.
\end{equation}
Here, each $U$ denotes a reversible circuit block, and $\dagger$ denotes its inverse operation. The circuit is read \textit{from right to left}. First, $U_{\rm load}$ applies $\mathcal{O}^{(k)}_{\sf load}$ to instantiate $v_u$ and $v_w$. Then, $U_{\wedge}$ applies bitwise AND to compute their overlap $v_u\wedge v_w$, $U_{\rm pop}$ computes the support count $\|v_u\wedge v_w\|_1$, and $U_{\rm cmp}$ compares this count with $\sigma$. The operation $X_{\rm flag}$ sets the target qubit $b$ to $1$ when the comparator output indicates that the threshold is reached.\footnote{$X_{\rm flag}$ is implemented using a Pauli-$X$ gate, which is the quantum analogue of a NOT gate. When the target qubit \(b\) is initialized to \(0\), the gate is applied only when the comparator condition is satisfied, i.e., when the support is at least \(\sigma\). In that case, it changes \(b\) from \(0\) to \(1\); otherwise, \(b\) remains \(0\). In Figure~\ref{fig:qsv_circuit}, the flag qubit is instead prepared as \(\ket{-}\). Under this preparation, applying the same conditional \(X_{\rm flag}\) to a satisfying candidate leaves the flag qubit in \(\ket{-}\) but attaches a negative sign to the candidate state, which is the marking used by the threshold-marking oracle.}
The inverse blocks $U_{\rm cmp}^{\dagger}$, $U_{\rm pop}^{\dagger}$, $U_{\wedge}^{\dagger}$, and $U_{\rm load}^{\dagger}$ then uncompute the temporary values and restore the workspace registers.
\end{definition}

\begin{figure}[t]
\centering
\small
\[
\Qcircuit @C=0.9em @R=0.8em {
 \lstick{\ket{u,w}_{\mathrm{idx}}} & \ctrl{1} & \qw      & \qw        & \qw            & \ctrl{1} & \qw \\
 \lstick{\ket{0}_{v_u}}            & \gate{\mathrm{load}\,v_u} & \ctrl{2} & \qw  & \qw & \gate{\mathrm{load}^{\dagger}} & \qw \\
 \lstick{\ket{0}_{v_w}}            & \gate{\mathrm{load}\,v_w} & \ctrl{1} & \qw  & \qw & \gate{\mathrm{load}^{\dagger}} & \qw \\
 \lstick{\ket{0}_{\wedge}}         & \qw      & \targ    & \multigate{1}{\mathrm{popcount}} & \ctrl{1} & \qw & \qw \\
 \lstick{\ket{0}_{s}}              & \qw      & \qw      & \ghost{\mathrm{popcount}} & \gate{\ge\sigma} & \qw & \qw \\
 \lstick{\ket{-}_{\mathrm{flag}}}  & \qw      & \qw      & \qw        & \ctrl{-1}      & \qw & \qw
}
\]
\caption{The threshold-marking oracle is implemented as a reversible circuit: Parent vectors are loaded, intersected ($\wedge$), \rB{the popcount is written to the support register $s$}, compared against $\sigma$, used to phase-mark the candidate, then all workspace is uncomputed \rB{by the inverse blocks}.}
\label{fig:qsv_circuit}
\end{figure}

\begin{example}
\label{ex:qsv_circuit_walkthrough}
Following Example~\ref{ex:load_oracle}, consider the candidate pair \((\text{B},\text{C})\) in the threshold-marking oracle. \rB{Intuitively, one oracle call answers a single yes/no question: ``do at least $\sigma$ transactions contain both B and C?''} With the index register set to \(\ket{\text{B},\text{C}}_{\rm idx}\) and the flag prepared in \(\ket{-}_{\rm flag}\)\rB{, the call proceeds in five steps.}
\rB{(i)~\emph{Load.}} The oracle first loads the two parent-vector registers:
\[
\ket{0^4}_{v_u}\ket{0^4}_{v_w}
\mapsto
\ket{0111}_{v_u}\ket{1110}_{v_w}.
\]
\rB{The circuit now holds, side by side, the lists of transactions containing B and containing C.} \rB{(ii)~\emph{Intersect.}} The bitwise-AND block then writes the overlap into the \(\wedge\) register:
\[
\ket{0111}_{v_u}\ket{1110}_{v_w}\ket{0^4}_{\wedge}
\mapsto
\ket{0111}_{v_u}\ket{1110}_{v_w}\ket{0110}_{\wedge},
\]
\rB{Each $1$ marks a transaction containing both parents, i.e., $T_2$ and $T_3$.} \rB{(iii)~\emph{Count.}} \rB{The} popcount block \rB{then} writes the support count into the \(s\) register:
\[
\ket{0110}_{\wedge}\ket{0}_{s}
\mapsto
\ket{0110}_{\wedge}\ket{2}_{s}.
\]
\rB{Here $s$ is exactly the support of $\{B,C\}$.} \rB{(iv)~\emph{Compare and mark.}} Since \(s=2\) meets the threshold \(\sigma=2\), the comparator and flag operation mark this candidate:
\[
\ket{2}_{s}\ket{-}_{\rm flag}
\mapsto
-\ket{2}_{s}\ket{-}_{\rm flag}.
\]
Because the flag is prepared in \(\ket{-}\), the flag register is not read as an output value. When the comparator condition is satisfied, the marking operation attaches a negative sign to the corresponding candidate state, allowing later amplification rounds to distinguish it from unmarked candidates.
\rB{(v)~\emph{Uncompute.}} Finally, the inverse blocks \(U_{\rm cmp}^{\dagger}\), \(U_{\rm pop}^{\dagger}\), \(U_{\wedge}^{\dagger}\), and \(U_{\rm load}^{\dagger}\) uncompute the workspace registers:
\[
\ket{0111}_{v_u}
\ket{1110}_{v_w}
\ket{0110}_{\wedge}
\ket{2}_{s}
\mapsto
\ket{0^4}_{v_u}
\ket{0^4}_{v_w}
\ket{0^4}_{\wedge}
\ket{0}_{s},
\]
leaving the candidate pair \((\text{B},\text{C})\) marked\rB{; 
the intermediate data are reset to their initial zero states, while the negative sign on the satisfying candidate state is preserved. This cleanup makes the oracle reusable in subsequent amplification rounds}. \rB{This completes one threshold-marking oracle call for the selected candidate pair. When the same oracle call is applied to the MAC \rB{s}uperposition from Example~\ref{ex:superposition}, steps (i)--(v) are applied to the candidate states represented in the index register. The oracle marks exactly those candidates whose support reaches \(\sigma\), i.e., \(\{A,C\}\), \(\{B,C\}\), \(\{B,E\}\), and \(\{C,E\}\).}
\exend
\end{example}

This oracle design enables \textit{Bit-Parallel Threshold Marking}. Rather than evaluating transactions sequentially, the bitwise AND $U_{\wedge}$ exploits full data parallelism to operate across all $M$ positions simultaneously. To achieve full parallelism, we structure $U_{\rm pop}$ as a \textit{parallel popcount accumulator} that utilizes a tree of reversible adders, and $U_{\rm cmp}$ is synthesized as a logarithmic-depth comparator. Consequently, after the initial instantiation, this design guarantees that each oracle call \(\mathcal{O}^{(k+1)}_{\ge\sigma}\) executes in logarithmic \rB{time} $O(\log |L_k|+\log M)$, where \(\log |L_k|\) corresponds to loading the parent vectors \(v_u\) and \(v_w\), and \(\log M\) corresponds to the popcount and comparator operations over \(M\) transactions, significantly outperforming the $\mathcal{O}(M)$ scan inherent to classical counting.

Consequently, QSV decouples state preparation from search execution. It establishes \rB{the} MAC superposition via CSI, drastically shrinking the exponential search space down to a highly targeted candidate subspace, and constructs the reusable threshold-marking oracle $\mathcal{O}^{(k+1)}_{\ge\sigma}$. Rather than applying this oracle immediately as a static filter, its invocation is deferred to the subsequent QIL module for itemset listing, to prevent redundant oracle executions during state preparation, effectively bounding the quantum register overhead and maintaining a highly optimized circuit depth.

\subsection{Quantum Itemset Listing (QIL)}
\label{subsec:QIl}

Given the MAC superposition prepared by QSV, QIL extracts the actual frequent $(k+1)$-itemsets through an iterative amplification process. Unlike classical filtering\rB{,} where a single evaluation isolates target records, a single invocation of the threshold-marking oracle $\mathcal{O}^{(k+1)}_{\ge\sigma}$ merely marks valid candidates without increasing their measurement probabilities; consequently, a static evaluation fails to retrieve the target frequent itemsets. To address this issue, QIL iteratively invokes $\mathcal{O}^{(k+1)}_{\ge\sigma}$ within an amplification loop,\footnote{This loop operates by periodically reflecting the probability amplitudes of the valid candidates around their mean. By combining this amplitude reflection with the oracle's marking, QIL effectively shifts probability mass from infrequent candidates into the valid ones.} which concentrates probability mass on the marked candidates, ensuring that a subsequent measurement collapses the quantum state into a verified frequent itemset with high probability.

To execute this amplification-based extraction reliably, QIL performs two supporting steps: \textit{boundary estimation} and \textit{dynamic deduplication}. 1)~Boundary estimation determines, or upper-bounds, how many candidates are currently marked by the threshold-marking oracle. This information is needed because amplitude amplification is periodic: applying too few amplification iterations may leave the marked candidates with insufficient measurement probability, whereas applying too many iterations may move the state past the high-probability measurement point. QIL evaluates this boundary by probing the amplification process induced by the same MAC \rB{s}uperposition and the same threshold-marking oracle \(\mathcal{O}^{(k+1)}_{\ge\sigma}\), either by extracting the periodic behavior of the amplification loop or by progressively increasing the number of oracle invocations until a suitable budget is reached. Crucially, because QSV decouples the preparation of \rB{the} MAC \rB{s}uperposition from candidate marking and extraction, boundary estimation reuses the same threshold-marking oracle \(\mathcal{O}^{(k+1)}_{\ge\sigma}\) without requiring circuit modifications. 2)~Dynamic deduplication excludes candidates that have already been returned in earlier extraction rounds. This step is needed because each amplification process followed by measurement produces only one candidate itemset; without exclusion, later rounds may repeatedly return the same itemset instead of discovering the remaining frequent candidates. After each measured candidate is verified, QIL records its candidate index and adds it to an exclusion condition. This condition is incorporated into the threshold-marking oracle, so previously extracted candidates are no longer marked in subsequent rounds. As a result, later amplification steps concentrate on the remaining undiscovered frequent itemsets.

\begin{example}
\label{ex:QIl}
Following Example~\ref{ex:superposition}, QIL operates on \rB{the} MAC \rB{s}uperposition over \(P_1\). Utilizing the threshold-marking oracle from Example~\ref{ex:qsv_circuit_walkthrough}, QIL estimates that \(4\) of the \(8\) states satisfy the support threshold and are marked. Thus, QIL uses \(4\) as the boundary estimate for the subsequent amplification loop. QIL repeatedly invokes \(\mathcal{O}^{(2)}_{\ge\sigma}\) to measure the index register.
Suppose the first measurement returns $\ket{\text{B},\text{C}}_{\rm idx}$, which corresponds to the candidate itemset \(\{\text{B},\text{C}\}\). QIL records this itemset and updates the exclusion condition so that \(\ket{\text{B},\text{C}}_{\rm idx}\) is not marked in subsequent rounds. Repeating this process until four distinct candidates are measured, yielding
\[
L_2 =
\{ \{\text{A},\text{C}\}, \{\text{B},\text{C}\}, \{\text{B},\text{E}\}, \{\text{C},\text{E}\} \}.
\]
Thus, QIL outputs \(L_2\) as the frequent 2-itemsets under the threshold \(\sigma=2\).
\exend
\end{example}
\subsection{Theoretical Analysis}
\label{sec:complexity}

This section analyzes the time complexity of QFM. 
Prior FIM studies have typically assessed practical miners empirically, rather than deriving stage-by-stage time complexity that separately accounts for preprocessing, candidate-domain construction, and frequent-itemset extraction~\cite{vo2012dbv,hamm2023tkde,martin2020ciclad,fang2024gpu,borgelt2023survey}, while the existing quantum attempt qARM provides query complexity only for the per-level support-checking step, leaving candidate generation running on traditional machines. In contrast, we derive the time complexity of QFM.
We reuse the notation of Section~\ref{sec:problem_formulation}; in particular, \(L_k\) and \(C_k\) denote the level-wise frequent set and candidate set, respectively.

With the level-wise bit-vector representation provided by QPR,\footnote{The derivation of QPR's time complexity is included in the proof of Theorem~\ref{thm:complexity_fim}.} we begin with QSV, which constructs the candidate domain and implements the threshold-marking oracle for the next mining level.

\begin{lemma}
\label{lem:qsv_depth}
\rB{The time complexity of one QSV threshold-marking oracle call at level $k+1$ is $\mathcal{O}(\log |L_k|+\log M)$.}
\end{lemma}

\begin{proof}
A threshold-marking oracle call first loads the parent vectors $v_u$ and $v_w$ from the level-$k$ cache,\footnote{\rB{The level-\(k\) cache is the table of materialized transaction-indicator vectors (Section~\ref{subsec:QFc}); the load oracle exposes it to the circuit as a reversible lookup, and addressing a table with \(S\) stored entries is realized by an \(\mathcal{O}(\log S)\)-depth selection circuit. This logarithmic addressing depth is the only cache cost charged in the bounds below.}} \rB{which requires $\mathcal{O}(\log |L_k|)$ time}. The bitwise AND $v_u\wedge v_w$ is bit-parallelizable and therefore contributes $\mathcal{O}(1)$ time. The support count is computed by a popcount circuit; using a parallel reduction tree followed by a parallel-prefix addition, this takes $\mathcal{O}(\log M)$ time. The comparator against $\sigma$ \rB{requires} $\mathcal{O}(\log M)$ time. \rB{The lemma follows.} \end{proof}

\begin{lemma}
\label{lem:qil_complexity}
The time complexity of QIL for level $k+1$ is
\[
\mathcal{O}\bigl(\sqrt{|C_{k+1}|\,\rB{\max\{1,|L_{k+1}|\}}}\,(\log |C_{k+1}|+\rB{\log |L_k|}+\log M)\bigr).
\]
\end{lemma}

\begin{proof}
QIL employs amplitude amplification~\cite{grover1996fast} to generate marked candidates. 
If \(|L_{k+1}|\) is not known in advance, QIL first obtains an upper bound using standard quantum counting~\cite{brassard2002quantum} or a doubling-style budget search over the same marking oracle. This estimation step and the subsequent listing iterations share the same MAC \rB{s}uperposition preparation and threshold-marking oracle. Following previous analyses of quantum counting and doubling search~\rB{\cite{brassard2002quantum,boyer1998tight}}, the estimation step adds only logarithmic-factor overhead, including the repetitions needed to raise the success probability to the bounded-error level. 
Specifically, listing \(|L_{k+1}|\) marked candidates from a candidate space of size \(|C_{k+1}|\) requires $\mathcal{O}\!\left(\sqrt{|C_{k+1}|\max\{1,|L_{k+1}|\}}\right)$
oracle iterations due to the following reasons. The term \(\max\{1,|L_{k+1}|\}\) covers the terminating level, where \(L_{k+1}=\emptyset\) and \(\mathcal{O}(\sqrt{|C_{k+1}|})\) oracle calls are required to decide that no marked candidate remains.
Each iteration consists of two parts. First, the non-oracle operations on the candidate index register require \(\mathcal{O}(\log |C_{k+1}|)\) time, including the exclusion-set check. This check can be implemented as a parallel comparator tree over at most \(|L_{k+1}|\) stored indices, and therefore its running time remains within the same logarithmic factor. Second, each iteration invokes the QSV threshold-marking oracle once, with \(\mathcal{O}(\log |L_k|+\log M)\) \rB{time} by Lemma~\ref{lem:qsv_depth}. The lemma follows.
\end{proof}

The above bound accounts for the process time of QIL; it remains to state the success guarantee under which the extracted candidates form valid frequent itemsets with high probability. The guarantee follows the standard bounded-error guarantee of amplitude amplification~\cite{grover1996fast,boyer1998tight}. With the iteration count determined from the boundary estimate, each extraction attempt returns a marked candidate with constant probability, while the boundary estimate follows the accuracy guarantee of quantum counting~\cite{brassard2002quantum}. After measurement, the returned candidate is verified directly from the parent vectors before being inserted into \(L_{k+1}\), so false positives are excluded from the output, ensuring that every reported itemset satisfies the support threshold. Repeating the extraction procedure reduces the remaining failure probability exponentially~\cite{nielsen2010quantum}, and the expected constant number of repetitions is absorbed in the complexity bound of Lemma~\ref{lem:qil_complexity}.

By combining the costs from all stages and summing over iterations from $k=1$ to $\hat{k}-1$, we derive the time complexity for FIM.

\begin{theorem}[Time complexity of QFM]
\label{thm:complexity_fim}
Let $B_{\hat{k}}=\sum_{k=1}^{\hat{k}-1}|C_{k+1}|$ denote the total number of level-wise candidates explored up to length $\hat{k}
\footnote{$B_{\hat{k}}$ is induced by the data and the threshold rather than being an input: it can be as small as $\binom{|L_1|}{2}=\mathcal{O}(N^2)$ when mining terminates at the second level (sparse data or high support); under practical support thresholds it typically stays polynomial, far below the exponential regime; it degenerates to $\Theta(2^N)$ only when $\sigma$ is so low that nearly every itemset is frequent. Please refer to~\onlineref{sec:comparison} for the detailed output-sensitive analysis.}
$\rG{, and let $Q_{\hat{k}}=\sum_{k=1}^{\hat{k}-1}\sqrt{|C_{k+1}|\,\max\{1,|L_{k+1}|\}}$ denote its output-sensitive refinement}. The time complexity of QFM is
\[
\rR{\mathcal{O}\!\left(NM+Q_{\hat{k}}(\log B_{\hat{k}}+\log M)\right).}
\]
\rR{Since $\max\{1,|L_{k+1}|\}\le|C_{k+1}|$ for every explored level, $Q_{\hat{k}}\le B_{\hat{k}}$, and the bound is in particular $\mathcal{O}(NM+B_{\hat{k}}(\log B_{\hat{k}}+\log M))$.}
\end{theorem}

\begin{proof}
QPR requires \(\mathcal{O}(NM)\) time by materializing the transaction-indicator vectors for all \(N\) items over \(M\) transactions. QPR then reduces the oracle-call count for identifying frequent 1-itemsets to \(\mathcal{O}(N\sqrt{M})\). Afterward, summing up the QIL time complexity in Lemma~\ref{lem:qil_complexity} gives
\[
\begin{aligned}
\sum_{k=1}^{\hat{k}-1}\mathcal{O}\Bigl(&\sqrt{|C_{k+1}|\,\max\{1,|L_{k+1}|\}}\\[-2pt]
&\;\cdot\bigl(\log |C_{k+1}|+\log |L_k|+\log M\bigr)\Bigr).
\end{aligned}
\]
Since $|L_k|$ is no larger than the corresponding level-wise search domain, both logarithmic factors are bounded by $\mathcal{O}(\log B_{\hat{k}})$. Therefore, the iterative part is bounded by $\mathcal{O}(Q_{\hat{k}}(\log B_{\hat{k}}+\log M))$. Adding the QPR cost yields the stated bound, and the relaxed form follows from $Q_{\hat{k}}\le B_{\hat{k}}$. The theorem follows. 
\end{proof}

\section{Experiments}
\label{sec:experiments}

\subsection{Experimental Setup}

\textbf{Datasets and Baselines.} 
We conduct experiments on eleven standard benchmarks for frequent itemset mining~\cite{dua2017uci,fournier2014spmf}. Table~\ref{tab:dataset-summary} summarizes the sizes of these datasets.
We compare QFM against six baselines: one quantum FIM algorithm, qARM~\cite{yu2022experimental}; four classical FIM or generalized itemset-mining algorithms, Eclat~\cite{zaki2000scalable}, FP-growth~\cite{han2000fp}, HAMM~\cite{hamm2023tkde}, 
and CICLAD~\cite{martin2020ciclad}; and one GPU-parallel algorithm, PHA-HUIM~\cite{fang2024gpu}.\footnote{HAMM and PHA-HUIM target high-utility itemset mining, a generalized itemset mining formulation. We instantiate their evaluation criterion with support-based filtering, making them directly applicable under the standard minimum-support FIM setting.}

\textbf{Metrics and Implementation.} 
Following~\cite{suchara2013overhead,litinski2019game,gidney2021factor,krol2024assessing}, we evaluate QFM using the following metrics: \textit{time}, \textit{logical depth}, and \textit{logical gate count}. We also conduct sensitivity analyses by varying the transaction ratio $\tau$ (percentage of transactions retained) and the minimum-support ratio $\sigma$ (defined with respect to the retained transactions). We further conduct ablation studies on QFM. We implement QFM in IBM Qiskit~\cite{QiskitCommunity2017} and Amazon Braket for our evaluation. Experimental results on Amazon Braket, quantum error measurements, peak logical-qubit footprint (i.e., the maximum logical-qubit usage in the circuit), and additional datasets are presented in~\onlineref{sec:additional_exp} due to space constraints. All experiments are conducted on an HPE ProLiant DL580 Gen9 server equipped with four 2.10 GHz Intel Xeon E7-8870 v4 CPUs, totaling 80 physical cores, and 1 TB of RAM.
The implementation and experiment scripts are available at \url{https://github.com/yxu616425/QFM_experiment}.

\begin{table}[t]
\centering
\caption{Summary of datasets. $M$ and $N$ denote the numbers of transactions and items, respectively.}
\label{tab:dataset-summary}
\setlength{\tabcolsep}{3.5pt}
\renewcommand{\arraystretch}{1.05}
\begin{tabular}{lrrrrr}
\toprule
Dataset & \textit{Mushroom} & \textit{Nursery} & \textit{Tic-Tac-Toe} & \textit{Car} & \textit{Connect-4} \\
\midrule
$M$ & 8124 & 12960 & 958 & 1728 & 67557 \\
$N$ & 22   & 8     & 10  & 6    & 42    \\
\bottomrule
\end{tabular}
\par\smallskip
\begin{tabular}{lrrrrrr}
\toprule
\rR{Dataset} & \textit{Accidents} & \textit{Pumsb} & \textit{Pumsb-star} & \textit{Chess} & \textit{Bike} & \textit{Skin} \\
\midrule
$M$ & 340183 & 49046 & 49046 & 3196 & 21078 & 245057 \\
$N$ & 468    & 2113  & 2088  & 75   & 67    & 12     \\
\bottomrule
\end{tabular}
\par\smallskip
\end{table}

\begin{table*}[t]
\centering
\footnotesize 
\renewcommand{\arraystretch}{0.88}
\setlength{\tabcolsep}{2.3pt}
\caption{Time \rB{(seconds)} comparison across transaction ratios.} 
\label{tab:baseline_tx}
\begin{tabular}{l|ccccc|ccccc}
\toprule
\textbf{Algo} & \multicolumn{5}{c|}{\textit{Mushroom}} & \multicolumn{5}{c}{\textit{Tic-Tac-Toe}} \\
& 30\% & 40\% & 50\% & 60\% & 70\% & 30\% & 40\% & 50\% & 60\% & 70\% \\
\midrule
\rR{QFM} & \textbf{2.09e-3} & \textbf{9.56e-3} & \textbf{0.06} & \textbf{0.07} & \textbf{0.07} & \textbf{3.47e-5} & \textbf{4.77e-5} & \textbf{6.63e-5} & \textbf{8.82e-5} & \textbf{1.12e-4} \\
qARM & 1.20 & 5.10 & 44.72 & 57.56 & 101.93 & 5.88e-3 & 0.01 & 0.02 & 0.04 & 0.08 \\
Eclat & 1.10 & 1.37 & 4.04 & 5.46 & 6.66 & 0.67 & 0.76 & 0.71 & 0.62 & 0.72 \\
FP-growth & 0.79 & 0.88 & 2.24 & 2.18 & 4.18 & 0.79 & 0.76 & 0.70 & 0.82 & 0.85 \\
HAMM & \underline{0.03} & \underline{0.13} & 1.65 & 1.94 & 3.15 & \underline{2.00e-3} & \underline{3.00e-3} & \underline{5.00e-3} & \underline{7.00e-3} & \underline{9.00e-3} \\
CICLAD & 21.95 & 44.48 & 51.98 & 71.23 & 85.59 & 0.18 & 0.29 & 0.30 & 0.43 & 0.43 \\
PHA-HUIM & 0.22 & 0.24 & \underline{0.27} & \underline{0.30} & \underline{0.32} & 0.13 & 0.14 & 0.13 & 0.13 & 0.13 \\
\midrule
\textbf{Algo} & \multicolumn{5}{c|}{\textit{Nursery}} & \multicolumn{5}{c}{\textit{Car}} \\
& 30\% & 40\% & 50\% & 60\% & 70\% & 30\% & 40\% & 50\% & 60\% & 70\% \\
\midrule
\rR{QFM} & \textbf{8.84e-5} & \textbf{1.38e-4} & \textbf{1.58e-4} & \textbf{1.74e-4} & \textbf{2.19e-4} & \textbf{2.49e-5} & \textbf{3.18e-5} & \textbf{3.81e-5} & \textbf{4.69e-5} & \textbf{5.27e-5} \\
qARM & \underline{7.00e-3} & \underline{0.01} & \underline{0.01} & \underline{0.03} & \underline{0.05} & 2.58e-3 & 7.06e-3 & 0.01 & 0.02 & 0.02 \\
Eclat & 0.75 & 0.78 & 0.80 & 0.78 & 0.77 & 0.72 & 0.81 & 0.62 & 0.79 & 0.76 \\
FP-growth & 1.15 & 2.17 & 2.28 & 0.76 & 0.92 & 0.82 & 0.78 & 0.68 & 0.71 & 0.81 \\
HAMM & 0.02 & 0.03 & 0.04 & \underline{0.03} & 0.06 & \underline{2.00e-3} & \underline{3.00e-3} & \underline{5.00e-3} & \underline{7.00e-3} & \underline{8.00e-3} \\
CICLAD & 7.35 & 9.95 & 14.86 & 19.47 & 24.23 & 0.05 & 0.08 & 0.11 & 0.16 & 0.19 \\
PHA-HUIM & 0.19 & 0.22 & 0.24 & 0.25 & 0.27 & 0.13 & 0.13 & 0.17 & 0.17 & 0.16 \\
\bottomrule
\end{tabular}
\end{table*}

\begin{table*}[t]
\centering
\footnotesize 
\renewcommand{\arraystretch}{0.88}
\setlength{\tabcolsep}{2.3pt}
\caption{Time \rB{(seconds)} comparison across minimum support thresholds.} 
\label{tab:baseline_ms}
\begin{tabular}{l|ccccc|ccccc}
\toprule
\textbf{Algo} & \multicolumn{5}{c|}{\textit{Mushroom}} & \multicolumn{5}{c}{\textit{Tic-Tac-Toe}} \\
& 10\% & 20\% & 30\% & 40\% & 50\% & 10\% & 20\% & 30\% & 40\% & 50\% \\
\midrule
\rR{QFM} & \textbf{0.12} & \textbf{8.88e-3} & \textbf{1.06e-3} & \textbf{6.64e-4} & \textbf{5.50e-4} & \textbf{1.70e-4} & \textbf{5.41e-5} & \textbf{4.52e-5} & \textbf{3.69e-5} & \textbf{3.69e-5} \\
qARM & 594.67 & 47.23 & 2.05 & 0.40 & 0.10 & 0.16 & 0.02 & 9.80e-3 & \underline{2.31e-3} & 1.39e-3 \\
Eclat & 30.05 & 8.59 & 1.75 & 1.29 & 0.92 & 0.84 & 0.60 & 0.62 & 0.63 & 0.64 \\
FP-growth & 1.73 & 1.28 & 0.99 & 0.86 & 0.79 & 0.79 & 0.82 & 0.72 & 0.64 & 0.72 \\
HAMM & 2.40 & 0.43 & \underline{0.09} & \underline{0.07} & \underline{0.07} & \underline{0.02} & \underline{0.01} & \underline{7.00e-3} & 4.00e-3 & \underline{1.00e-3} \\
CICLAD & 235.23 & 235.23 & 147.67 & 147.67 & 147.67 & 1.02 & 1.02 & 1.02 & 1.02 & 1.02 \\
PHA-HUIM & \underline{0.43} & \underline{0.40} & 0.37 & 0.35 & 0.33 & 0.16 & 0.15 & 0.15 & 0.16 & 0.14 \\
\midrule
\textbf{Algo} & \multicolumn{5}{c|}{\textit{Nursery}} & \multicolumn{5}{c}{\textit{Car}} \\
& 10\% & 20\% & 30\% & 40\% & 50\% & 10\% & 20\% & 30\% & 40\% & 50\% \\
\midrule
\rR{QFM} & \textbf{2.88e-4} & \textbf{1.96e-4} & \textbf{1.94e-4} & \textbf{1.59e-4} & \textbf{1.59e-4} & \textbf{7.84e-5} & \textbf{5.64e-5} & \textbf{4.99e-5} & \textbf{4.37e-5} & \textbf{4.37e-5} \\
qARM & 0.11 & \underline{0.02} & \underline{0.02} & \underline{0.01} & \underline{0.01} & 0.04 & \underline{0.01} & \underline{5.12e-3} & \underline{1.69e-3} & \underline{1.69e-3} \\
Eclat & 0.89 & 0.73 & 0.70 & 0.81 & 0.73 & 0.68 & 0.64 & 0.58 & 0.63 & 0.65 \\
FP-growth & 0.93 & 0.86 & 0.68 & 0.75 & 0.70 & 0.73 & 0.80 & 0.64 & 0.67 & 0.60 \\
HAMM & \underline{0.09} & 0.08 & 0.06 & 0.06 & 0.05 & \underline{0.01} & \underline{0.01} & 7.00e-3 & 6.00e-3 & 6.00e-3 \\
CICLAD & 40.56 & 40.56 & 40.56 & 40.56 & 40.56 & 0.47 & 0.47 & 0.47 & 0.47 & 0.47 \\
PHA-HUIM & 0.34 & 0.33 & 0.30 & 0.28 & 0.28 & 0.15 & 0.16 & 0.15 & 0.15 & 0.16 \\
\bottomrule
\end{tabular}
\end{table*}

\subsection{Quantitative Results}

Table~\ref{tab:baseline_tx} reports time under varied transaction ratios $\tau$. For all datasets, QFM consistently outperforms all baselines. As $\tau$ increases, the overall time generally grows due to the enlarged transaction database. Nevertheless, QFM exhibits only mild growth because its support verification relies on the threshold-marking oracle with logarithmic depth, preventing the verification cost from scaling linearly with the transaction count. In contrast, the classical baselines show pronounced growth, especially on denser datasets such as \textit{Mushroom} and \textit{Nursery}. Dense datasets intrinsically enlarge the candidate space and intensify repeated support-checking overhead, making the cost of classical mining procedures increase rapidly due to their reliance on sequential database scans. Similarly, PHA-HUIM exhibits mixed trends across datasets, since its GPU performance remains constrained by classical data movement bottlenecks, irregular memory access, and repeated verification costs, which are particularly visible on \textit{Nursery}. Finally, qARM remains severely limited since it merely delegates support evaluation to a quantum subroutine while leaving candidate generation to classical machine bottlenecks. Consequently, it avoids severe performance degradation only under extremely small $\tau$ where the search space is trivially simple. As the transaction volume grows, its hybrid overhead quickly dominates, rendering it uncompetitive. This contrast is most evident on \textit{Nursery}. While all baselines are burdened by dense transaction structures and repeated verification, QFM leverages \rB{the} MAC \rB{s}uperposition to substantially compress the search space and applies its logarithmic-depth threshold-marking oracle to verify support efficiently, jointly enabling QFM to maintain a clear and robust advantage as $\tau$ increases.

Table~\ref{tab:baseline_ms} reports time under varied minimum-support ratios $\sigma$. Across all datasets, QFM consistently achieves the best performance. As $\sigma$ increases, the mining workload naturally decreases due to tighter support constraints. For QFM, a higher $\sigma$ sharply reduces the number of valid parent itemsets, allowing \rB{the} MAC \rB{s}uperposition to instantiate a significantly more compact CSI and minimize state preparation overhead. Furthermore, when fewer candidates satisfy the support threshold, QIL needs fewer quantum amplification rounds to enumerate them. Consequently, QFM maintains a remarkably low and stable execution time as $\sigma$ grows, exhibiting a particularly sharp drop on \textit{Mushroom}. While the classical and GPU baselines also experience reduced workloads at higher $\sigma$, they fail to scale down their execution time as effectively. Their performance remains bottlenecked by the inherent overhead of explicit candidate enumeration and sequential memory traversals, which persist even when the valid itemset count shrinks. CICLAD presents a notable exception to this trend, showing little sensitivity to $\sigma$, as its runtime is almost entirely dominated by classical stream-processing initialization rather than the mining process itself. Finally, the performance of qARM remains limited because it still delegates candidate generation to the classical CPU, allowing it to avoid severe degradation only in extreme high-support scenarios where the search space is trivially sparse.

\begin{figure}[t]
\centering
{\revisioncolor{yellow!50!black}%
\begin{tikzpicture}
\begin{axis}[
    ymode=log,
    log basis y={10},
    width=\columnwidth,
    height=5.2cm,
    ymin=0.05,
    ymax=120000,
    ylabel={Time (s, log scale)},
    xmin=-0.75,
    xmax=8.95, 
    xtick={0,2.6,5.2,7.8},
    xticklabels={Accidents,Connect-4,Pumsb,Chess},
    xticklabel style={font=\scriptsize, rotate=12, anchor=east, yshift=-2mm},
    ymajorgrids=true,
    grid style={dashed,gray!35},
    legend style={at={(0.5,1.12)}, anchor=south, legend columns=3, draw=none, font=\scriptsize},
    tick label style={font=\scriptsize},
    label style={font=\scriptsize},
    axis on top=false,
]
\addlegendimage{area legend, draw=black, fill=black!8}\addlegendentry{QFM}
\addlegendimage{area legend, draw=black, fill=black!28}\addlegendentry{qARM}
\addlegendimage{area legend, draw=black, fill=black!45}\addlegendentry{FP-growth}
\addlegendimage{area legend, draw=black, fill=black!60}\addlegendentry{HAMM}
\addlegendimage{area legend, draw=black, fill=black!75}\addlegendentry{Eclat}
\addlegendimage{area legend, draw=rkR, fill=black!92}\addlegendentry{PHA-HUIM} 

\def\figbase{0.05}
\newcommand{\timebar}[4]{\filldraw[fill=#4, draw=black, line width=0.30pt] (axis cs:#1,\figbase) rectangle (axis cs:#2,#3);}

\draw[densely dashed, gray!38] (axis cs:1.30,\figbase) -- (axis cs:1.30,120000);
\draw[densely dashed, gray!38] (axis cs:3.90,\figbase) -- (axis cs:3.90,120000);
\draw[densely dashed, gray!38] (axis cs:6.50,\figbase) -- (axis cs:6.50,120000);

\timebar{-0.50}{-0.34}{51.71}{black!8}
\timebar{-0.25}{-0.09}{24162.49}{black!28}
\timebar{0.00}{0.16}{169.12}{black!45}
\timebar{0.25}{0.41}{197.53}{black!60}
\timebar{0.50}{0.66}{30525.81}{black!75}
\filldraw[fill=black!92, draw=rkR, line width=0.5pt] (axis cs:0.75,\figbase) rectangle (axis cs:0.91,147.51); 

\timebar{2.10}{2.26}{11.43}{black!8}
\timebar{2.35}{2.51}{5481.65}{black!28}
\timebar{2.60}{2.76}{45.69}{black!45}
\timebar{2.85}{3.01}{31.60}{black!60}
\timebar{3.10}{3.26}{7509.78}{black!75}
\filldraw[fill=black!92, draw=rkR, line width=0.5pt] (axis cs:3.35,\figbase) rectangle (axis cs:3.51,22.417); 

\timebar{4.70}{4.86}{61.62}{black!8}
\timebar{4.95}{5.11}{37464.10}{black!28}
\timebar{5.20}{5.36}{206.04}{black!45}
\timebar{5.45}{5.61}{202.93}{black!60}
\timebar{5.70}{5.86}{36112.53}{black!75}
\filldraw[fill=black!92, draw=rkR, line width=0.5pt] (axis cs:5.95,\figbase) rectangle (axis cs:6.11,98); 

\timebar{7.30}{7.46}{0.105}{black!8}
\timebar{7.55}{7.71}{38.03}{black!28}
\timebar{7.80}{7.96}{1.329}{black!45}
\timebar{8.05}{8.21}{2.283}{black!60}
\timebar{8.30}{8.46}{6.124}{black!75}
\filldraw[fill=black!92, draw=rkR, line width=0.5pt] (axis cs:8.55,\figbase) rectangle (axis cs:8.71,0.439); 
\end{axis}
\end{tikzpicture}%
}

\caption{Time comparison on various large datasets.}
\label{fig:large_time_rebuttal}
\end{figure}
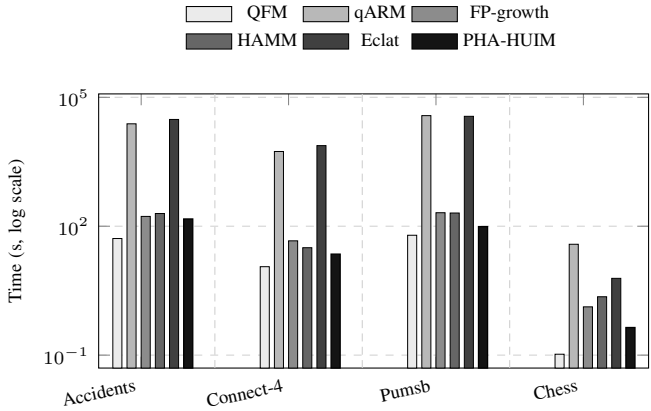

\begin{table}[t]
\centering
\setlength{\tabcolsep}{3pt}
\caption{Quantum resource usage comparison. 
\rB{The qARM compilations exceed the per-circuit gate limit of the execution platform (Amazon Braket); the affected compiled metrics are reported as truncated lower bounds ($>$30000).}}
\label{tab:yu_quantum_resource}
\begin{tabular}{lrrrr}
\toprule
\multirow{2}{*}{Method}
& \multicolumn{2}{c}{\textit{Bike}}
& \multicolumn{2}{c}{\textit{Skin}} \\
\cmidrule(lr){2-3} \cmidrule(lr){4-5}
& Depth & Gate count
& Depth & Gate count \\
\midrule
QFM  & 1460       & 14760       & 1820       & 13720 \\
qARM & $>30000$ & $>30000$ & $>30000$ & $>30000$ \\
\bottomrule
\end{tabular}
\end{table}

\begin{table*}[t]
\centering
\small 
\renewcommand{\arraystretch}{0.95}
\setlength{\tabcolsep}{4pt}
\caption[Ablation study.]{Ablation study \rR{(seconds; upper block varies $\tau$, lower block varies $\sigma$).}} 
\label{tab:ablation_tx}\label{tab:ablation_ms}
\begin{tabular}{l|ccccc|ccccc}
\toprule
\multicolumn{11}{c}{\rR{(a) Varying transaction ratio $\tau$}} \\
\midrule
\textbf{Variant} & \multicolumn{5}{c|}{\textit{Mushroom}} & \multicolumn{5}{c}{\textit{Nursery}} \\
& 30\% & 40\% & 50\% & 60\% & 70\% & 30\% & 40\% & 50\% & 60\% & 70\% \\
\midrule
\rR{QFM} & \textbf{2.09e-3} & \textbf{9.56e-3} & \textbf{0.06} & \textbf{0.07} & \textbf{0.07}  & \textbf{8.84e-5} & \textbf{1.38e-4} & \textbf{1.58e-4} & \textbf{1.74e-4} & \textbf{2.19e-4}  \\
\rR{w/o QIL} & \underline{2.28e-3} & 0.01 & \underline{0.07} & \underline{0.08} & \underline{0.08}  & \underline{8.85e-5} & \underline{1.82e-4} & \underline{2.16e-4} & \underline{2.41e-4} & \underline{3.01e-4}  \\
\rR{w/o QPR} & 2.38e-3 & \underline{9.99e-3} & \textbf{0.06} & \textbf{0.07} & \textbf{0.07}  & 3.02e-4 & 4.41e-4 & 5.52e-4 & 6.62e-4 & 8.06e-4  \\
\rR{w/o \rB{QSV}} & 0.03 & 0.13 & 0.65 & 0.82 & 0.85  & 9.50e-5 & 3.57e-4 & 4.29e-4 & 4.85e-4 & 8.46e-4  \\
\midrule
\multicolumn{11}{c}{\rR{(b) Varying minimum support threshold $\sigma$}} \\
\midrule
\textbf{Variant} & \multicolumn{5}{c|}{\textit{Mushroom}} & \multicolumn{5}{c}{\textit{Nursery}} \\
& 10\% & 20\% & 30\% & 40\% & 50\% & 10\% & 20\% & 30\% & 40\% & 50\% \\
\midrule
\rR{QFM} & \textbf{0.12} & \textbf{8.88e-3} & \textbf{1.06e-3} & \textbf{6.64e-4} & \textbf{5.50e-4}  & \textbf{2.88e-4} & \textbf{1.96e-4} & \textbf{1.94e-4} & \textbf{1.59e-4} & \textbf{1.59e-4}  \\
\rR{w/o QIL} & \underline{0.13} & \underline{9.49e-3} & \underline{1.14e-3} & \underline{6.88e-4} & \underline{5.55e-4}  & \underline{4.08e-4} & 2.77e-4 & 2.16e-4 & \textbf{1.59e-4} & \textbf{1.59e-4}  \\
\rR{w/o QPR} & \textbf{0.12} & 0.01 & 2.99e-3 & 2.60e-3 & 2.48e-3  & 1.17e-3 & 1.07e-3 & 1.07e-3 & \underline{1.04e-3} & \underline{1.04e-3}  \\
\rR{w/o \rB{QSV}} & 1.36 & 0.12 & 8.90e-3 & 2.66e-3 & 1.08e-3  & 1.46e-3 & \underline{2.43e-4} & \underline{2.07e-4} & \textbf{1.59e-4} & \textbf{1.59e-4}  \\
\bottomrule
\end{tabular}
\end{table*}

Figure~\ref{fig:large_time_rebuttal} demonstrates the running time on larger benchmark datasets.\footnote{CICLAD is omitted because it ran for more than 24 hours.} Consistent with the results in Tables~\ref{tab:baseline_tx} and~\ref{tab:baseline_ms}, QFM continues to outperform qARM and the classical and GPU baselines across these datasets. Even as the database size increases, \rB{the} MAC \rB{s}uperposition remains effective by transforming the search space into a contiguous register via CSI, circumventing the deep conditional logic needed for vast itemset candidates. Furthermore, the logarithmic-depth threshold-marking oracle ensures sustained efficiency through the parallel popcount accumulator, which evaluates the support threshold across all transactions simultaneously.

\subsection{Quantum Resource Usage Comparison}
Table~\ref{tab:yu_quantum_resource} further compares the quantum resource usage of QFM and qARM on \textit{Bike} and \textit{Skin}. QFM substantially reduces both logical depth and logical gate count, indicating lower sequential and total logical operation costs. In particular, QFM remains below 2000 in depth and below 15000 in gate count, whereas qARM exceeds 30000 in both metrics. This reduction comes from the designs of QFM. \rB{The} MAC \rB{s}uperposition avoids constructing complex circuits over the search space, while the logarithmic-depth threshold-marking oracle keeps support verification shallow through bitwise-AND, parallel popcount accumulation, and logarithmic-depth comparison. As a result, QFM produces quantum circuits that are both shallower and lighter than those of qARM.

\subsection{Ablation Study}

To assess the contribution of each quantum module in QFM, we compare QFM with three ablated variants: 1)~\textbf{w/o QPR}, which replaces QPR with a classical linear scan and bypasses Bit-Vector Qubit Encoding; 2)~\textbf{w/o QSV}, which removes QSV and embeds support verification directly into the amplification loop; and 3)~\textbf{w/o QIL}, which replaces QIL with sequential classical listing without amplitude amplification.

The upper block of Table~\ref{tab:ablation_tx} reports time under varied transaction ratios $\tau$. Across all datasets, QFM consistently achieves the lowest time among the variants, showing that the three quantum modules jointly contribute to its efficiency. Among the variants, \textbf{w/o QSV} incurs the largest time, followed by \textbf{w/o QPR}, indicating that QSV is the most critical component, since it cleanly decouples candidate state preparation from support verification by constructing both \rB{the} MAC \rB{s}uperposition and the bit-parallel threshold-marking oracle. As $\tau$ increases, \textbf{w/o QSV} grows the fastest because support verification is no longer handled by the logarithmic-depth threshold-marking oracle; instead, the amplification loop inherits a linear-cost counting procedure. This confirms the role of QSV in reducing support verification to logarithmic depth with respect to the transaction count. \textbf{w/o QPR} also reports larger values on complex datasets such as \textit{Nursery}, since removing QPR eliminates early pruning and compact Bit-Vector Qubit Encoding before the candidate domain is prepared.

The lower block of Table~\ref{tab:ablation_ms} reports time under varied minimum-support ratios $\sigma$. The impact of removing QSV is most visible in the low-support regime, where support verification is invoked more frequently due to the larger candidate domain. In these settings, QFM remains efficient because QIL exploits the logarithmic-depth threshold-marking oracle prepared by QSV for support verification, whereas \textbf{w/o QSV} pays a higher verification cost per iteration. Similarly, the performance degradation of \textbf{w/o QIL} is more pronounced when $\sigma$ is small, because more frequent itemsets need to be extracted. This confirms the role of QIL in efficient pattern listing. Overall, the ablation results show that QPR, QSV, and QIL play distinct but complementary roles, whose combination enables QFM to achieve its full performance benefit.

\section{Conclusion}

In this paper, we explore Frequent Itemset Mining in the quantum-computing setting and propose \textit{\fullalgo\ (QFM)}. We first identify three quantum-specific challenges, including reversible data-access mismatch, candidate-domain encoding gap, and repeated reversible support-verification cost. To address these challenges, QFM orchestrates Bit-Vector Qubit Encoding, Mining-Aware Candidate Superposition, and Bit-Parallel Threshold Marking within a unified quantum mining framework. We further provide theoretical analysis in terms of time complexity and validate its implementation on IBM Qiskit and Amazon Braket. Extensive experiments on real-world benchmark datasets show that QFM consistently outperforms representative quantum, classical, and GPU baselines, demonstrating its effectiveness for frequent itemset mining. Future work will investigate further circuit-level optimizations for QFM and extend the framework to related mining tasks such as high-utility and uncertain itemset mining.

\clearpage
\section*{AI-Generated Content Acknowledgement}
ChatGPT (OpenAI), Claude Code (Anthropic), and Gemini (Google) were used during manuscript preparation for language editing, LaTeX reformatting, terminology consistency checks, and assistance in debugging or checking experiment scripts associated with Section~V. The authors reviewed and verified all technical statements, mathematical derivations, experimental procedures, tables, and conclusions and remain fully responsible for the submitted content.

\bibliographystyle{IEEEtran}
\bibliography{listing}

@article{litinski2019game,
  title   = {A Game of Surface Codes: Large-Scale Quantum Computing with Lattice Surgery},
  author  = {Daniel Litinski},
  journal = {Quantum},
  volume  = {3},
  pages   = {128},
  year    = {2019},
  doi     = {10.22331/q-2019-03-05-128}
}

@article{gidney2021factor,
  title   = {How to factor 2048 bit {RSA} integers in 8 hours using 20 million noisy qubits},
  author  = {Craig Gidney and Martin Eker{\aa}},
  journal = {Quantum},
  volume  = {5},
  pages   = {433},
  year    = {2021},
  doi     = {10.22331/q-2021-04-15-433}
}

@misc{krol2024assessing,
  title         = {Assessing the Requirements for Industry Relevant Quantum Computation},
  author        = {Anna M. Krol and Marvin Erdmann and Ewan Munro and Andre Luckow and Zaid Al-Ars},
  year          = {2024},
  eprint        = {2408.02587},
  archivePrefix = {arXiv},
  primaryClass  = {quant-ph},
  url           = {https://arxiv.org/abs/2408.02587}
}

@article{preskill2018nisq,
  author  = {John Preskill},
  title   = {Quantum Computing in the {NISQ} era and beyond},
  journal = {Quantum},
  volume  = {2},
  pages   = {79},
  year    = {2018},
  month   = aug,
  doi     = {10.22331/q-2018-08-06-79},
  url     = {https://doi.org/10.22331/q-2018-08-06-79}
}

@inproceedings{martin2020ciclad,
  author    = {Tomas Martin and Guy Francoeur and Petko Valtchev},
  title     = {CICLAD: A Fast and Memory-efficient Closed Itemset Miner for Streams},
  booktitle = {Proceedings of the 26th ACM SIGKDD International Conference on Knowledge Discovery \& Data Mining (KDD '20)},
  pages     = {1810--1818},
  year      = {2020},
  doi       = {10.1145/3394486.3403232}
}

@article{suchara2013overhead,
  author  = {Martin Suchara and Arvin Faruque and Ching{-}Yi Lai and Gerardo Paz and Frederic T. Chong and John Kubiatowicz},
  title   = {Comparing the Overhead of Topological and Concatenated Quantum Error Correction},
  journal = {arXiv preprint arXiv:1312.2316},
  year    = {2013},
  url     = {https://arxiv.org/abs/1312.2316}
}

@article{zaki2000scalable,
  author  = {Mohammed J. Zaki},
  title   = {Scalable Algorithms for Association Mining},
  journal = {IEEE Transactions on Knowledge and Data Engineering},
  volume  = {12},
  number  = {3},
  pages   = {372--390},
  year    = {2000},
  doi     = {10.1109/69.846291}
}

@inproceedings{han2000fp,
  author    = {Jiawei Han and Jian Pei and Yiwen Yin},
  title     = {Mining Frequent Patterns without Candidate Generation},
  booktitle = {Proceedings of the 2000 ACM SIGMOD International Conference on Management of Data (SIGMOD '00)},
  year      = {2000},
  pages     = {1--12},
  doi       = {10.1145/342009.335372}
}

@inproceedings{Agrawal94,
  author    = {Rakesh Agrawal and Ramakrishnan Srikant},
  title     = {Fast Algorithms for Mining Association Rules},
  booktitle = {Proceedings of the 20th International Conference on Very Large Data Bases (VLDB '94)},
  year      = {1994},
  pages     = {487--499},
  address   = {Santiago, Chile},
  publisher = {Morgan Kaufmann}
}

@article{FinanceReview2018,
  author = {Rom{\'a}n Or{\'u}s and Samuel Mugel and Enrique Lizaso},
  title = {Quantum Computing for Finance: Overview and Prospects},
  journal = {Reviews in Physics},
  volume = {4},
  pages = {100028},
  year = {2019},
  doi = {10.1016/j.revip.2019.100028}
}

@misc{QiskitCommunity2017,
  title = {Qiskit: An Open-Source Framework for Quantum Computing},
  author = {{Qiskit Community}},
  year = {2017},
  doi = {10.5281/zenodo.2562110},
  url = {https://github.com/Qiskit/qiskit}
}

@inproceedings{moens2013frequent,
  author    = {Sandy Moens and Emin Aksehirli and Bart Goethals},
  title     = {Frequent Itemset Mining for Big Data},
  booktitle = {2013 IEEE International Conference on Big Data},
  year      = {2013},
  pages     = {111--118},
  doi       = {10.1109/BigData.2013.6691742}
}

@inproceedings{ambainis2004quantum,
  title = {Quantum query algorithms and lower bounds},
  author = {Ambainis, Andris},
  booktitle = {Classical and New Paradigms of Computation and their Complexity Hierarchies},
  pages = {15--32},
  year = {2004},
  organization = {Springer}
}

@article{arute2019quantum,
  author  = {Frank Arute \emph{et al.}},
  title   = {Quantum Supremacy Using a Programmable Superconducting Processor},
  journal = {Nature},
  volume  = {574},
  pages   = {505--510},
  year    = {2019},
  doi     = {10.1038/s41586-019-1666-5}
}

@article{biamonte2017quantum,
  author  = {Biamonte, Jacob and Wittek, Peter and Pancotti, Nicola and Rebentrost, Patrick and Wiebe, Nathan and Lloyd, Seth},
  title   = {Quantum machine learning},
  journal = {Nature},
  volume  = {549},
  number  = {7671},
  pages   = {195--202},
  year    = {2017},
  doi     = {10.1038/nature23474}
}

@article{borgelt2023survey,
  author  = {Borgelt, Christian},
  title   = {Frequent item set mining},
  journal = {Wiley Interdisciplinary Reviews: Data Mining and Knowledge Discovery},
  volume  = {2},
  number  = {6},
  pages   = {437--456},
  year    = {2012},
  doi     = {10.1002/widm.1074}
}

@incollection{brassard2002quantum,
  author    = {Gilles Brassard and Peter H{\o}yer and Michele Mosca and Alain Tapp},
  title     = {Quantum Amplitude Amplification and Estimation},
  booktitle = {Quantum Computation and Information},
  series    = {Contemporary Mathematics},
  volume    = {305},
  pages     = {53--74},
  publisher = {American Mathematical Society},
  year      = {2002},
  doi       = {10.1090/conm/305/05215}
}

@article{buhrman2002complexity,
  title = {Complexity measures and decision tree complexity: a survey},
  author = {Buhrman, Harry and De Wolf, Ronald},
  journal = {Theoretical Computer Science},
  volume = {288},
  number = {1},
  pages = {21--43},
  year = {2002},
  publisher = {Elsevier}
}

@article{chen2002fpmax,
  title = {FPMax: Mining maximal frequent patterns without candidate generation},
  author = {Chen, Wen and others},
  journal = {IEEE Transactions on Knowledge and Data Engineering},
  volume = {14},
  number = {5},
  pages = {1002--1016},
  year = {2002},
  publisher = {IEEE}
}

@article{deutsch1985quantum,
  author  = {Deutsch, David},
  title   = {Quantum theory, the {Church--Turing} principle and the universal quantum computer},
  journal = {Proceedings of the Royal Society of London. Series A, Mathematical and Physical Sciences},
  volume  = {400},
  number  = {1818},
  pages   = {97--117},
  year    = {1985},
  doi     = {10.1098/rspa.1985.0070}
}

@misc{dong2023fpstreampp,
  author = {Giannella, Chris and Han, Jiawei and Pei, Jian and Yan, Xifeng and Yu, Philip S.},
  title  = {Mining Frequent Patterns in Data Streams at Multiple Time Granularities},
  year   = {2003},
  url    = {https://hanj.cs.illinois.edu/pdf/fpstm03.pdf}
}

@misc{dua2017uci,
  author       = {Dheeru Dua and Casey Graff},
  title        = {{UCI} Machine Learning Repository},
  howpublished = {\url{https://archive.ics.uci.edu/}},
  year         = {2017}
}

@inproceedings{duong2021gctree,
  author    = {Chen, Junbo and Li, Shanping},
  title     = {GC-tree: A Fast Online Algorithm for Mining Frequent Closed Itemsets},
  booktitle = {Advances in Knowledge Discovery and Data Mining (PAKDD 2007)},
  series    = {Lecture Notes in Computer Science},
  volume    = {4426},
  pages     = {457--468},
  year      = {2007},
  publisher = {Springer},
  doi       = {10.1007/978-3-540-77018-3_45}
}

@misc{farhi2014qaoa,
  author       = {Edward Farhi and Jeffrey Goldstone and Sam Gutmann},
  title        = {A Quantum Approximate Optimization Algorithm},
  howpublished = {arXiv:1411.4028},
  year         = {2014}
}

@article{feynman1982simulating,
  author = {Feynman, Richard P.},
  title = {Simulating physics with computers},
  journal = {International Journal of Theoretical Physics},
  volume = {21},
  number = {6-7},
  pages = {467-482},
  year = {1982}
}

@inproceedings{grover1996fast,
  author    = {Grover, Lov K.},
  title     = {A Fast Quantum Mechanical Algorithm for Database Search},
  booktitle = {Proceedings of the 28th Annual ACM Symposium on Theory of Computing (STOC '96)},
  year      = {1996},
  pages     = {212--219},
  publisher = {ACM},
  doi       = {10.1145/237814.237866}
}

@article{hamm2023tkde,
  author  = {Qu, Jianfeng and Fournier-Viger, Philippe and Liu, Mengjia and Hang, Bing and Hu, Cheng},
  title   = {Mining High Utility Itemsets Using Prefix Trees and Utility Vectors},
  journal = {IEEE Transactions on Knowledge and Data Engineering},
  volume  = {35},
  number  = {10},
  pages   = {10224--10236},
  year    = {2023},
  doi     = {10.1109/TKDE.2023.3256126}
}

@article{han2004fptree,
  author  = {Jiawei Han and Jian Pei and Yiwen Yin and Runying Mao},
  title   = {Mining Frequent Patterns without Candidate Generation: A Frequent-Pattern Tree Approach},
  journal = {Data Mining and Knowledge Discovery},
  volume  = {8},
  number  = {1},
  pages   = {53--87},
  year    = {2004}
}

@misc{ibmroadmap2024,
  author = {{IBM Quantum}},
  title  = {{IBM Quantum Roadmap}},
  year   = {2024},
  url    = {https://www.ibm.com/roadmaps/quantum/},
  note   = {Accessed 13 May 2025}
}

@article{iqbal2021tkfim,
  author  = {Iqbal, Saood and Shahid, Abdul and Roman, Muhammad and Khan, Zahid and Al-Otaibi, Shaha and Yu, Lisu},
  title   = {Tk-FIM: Top-k Frequent Itemset Mining with Deep Learning},
  journal = {PeerJ Computer Science},
  volume  = {7},
  pages   = {e717},
  year    = {2021},
  doi     = {10.7717/peerj-cs.717}
}

@article{kandala2017vqe,
  author  = {Abhinav Kandala \emph{et al.}},
  title   = {Hardware-Efficient Variational Quantum Eigensolver for Small Molecules and Quantum Magnets},
  journal = {Nature},
  volume  = {549},
  pages   = {242--246},
  year    = {2017},
  doi     = {10.1038/nature23879}
}

@inproceedings{li2008pfp,
  author    = {Haoyuan Li and Yi Wang and Dong Zhang and Ming Zhang and Edward Y. Chang},
  title     = {Parallel {FP}-Growth for Query Recommendation},
  booktitle = {Proceedings of the 2008 ACM Conference on Recommender Systems (RecSys '08)},
  year      = {2008},
  pages     = {107--114},
  doi       = {10.1145/1454008.1454027}
}

@inproceedings{liu2012fhm,
  author    = {Shuai-Ping Liu and Cheng-Wei Wu and Vincent S. Tseng},
  title     = {{FHM}: Faster High-Utility Itemset Mining Using Estimated Utility Co-occurrence Pruning},
  booktitle = {Proceedings of the 18th ACM SIGKDD International Conference on Knowledge Discovery and Data Mining},
  year      = {2012},
  pages     = {883--892}
}

@article{lucchese2005fast,
  title = {Fast and memory efficient mining of frequent closed itemsets},
  author = {Lucchese, Claudio and Orlando, Salvatore and Perego, Raffaele},
  journal = {IEEE Transactions on Knowledge and Data Engineering},
  volume = {18},
  number = {1},
  pages = {21--36},
  year = {2005},
  publisher = {IEEE}
}

@article{madsen2022borealis,
  author  = {Madsen, Lars S. and Laudenbach, Fabian and Falamarzi Askarani, Mohsen and others},
  title   = {Quantum computational advantage with a programmable photonic processor},
  journal = {Nature},
  volume  = {606},
  number  = {7912},
  pages   = {75--81},
  year    = {2022},
  doi     = {10.1038/s41586-022-04725-x}
}

@article{nguyen2020satfim,
  author  = {Henriques, Rui and Lynce, In{\^e}s and Manquinho, Vasco M.},
  title   = {On When and How to Use SAT to Mine Frequent Itemsets},
  journal = {arXiv preprint arXiv:1207.6253},
  year    = {2012},
  url     = {https://arxiv.org/abs/1207.6253}
}

@inproceedings{pasquier1999discovering,
  title = {Discovering frequent closed itemsets for association rules},
  author = {Pasquier, Nicolas and Bastide, Yves and Taouil, Rafik and Lakhal, Lotfi},
  booktitle = {Database Theory---ICDT'99: 7th International Conference Jerusalem, Israel, January 10--12, 1999 Proceedings 7},
  pages = {398--416},
  year = {1999},
  organization = {Springer}
}

@inproceedings{pei2001hmine,
  title = {H-Mine: Hyper-structure mining of frequent patterns in large databases},
  author = {Pei, Jian and others},
  booktitle = {Proceedings 2001 IEEE International Conference on Data Mining},
  pages = {441--448},
  year = {2001},
  organization = {IEEE}
}

@article{raj2020eafim,
  title = {EAFIM: efficient apriori-based frequent itemset mining algorithm on Spark for big transactional data},
  author = {Raj, Shashi and Ramesh, Dharavath and Sreenu, M and Sethi, Krishan Kumar},
  journal = {Knowledge and Information Systems},
  volume = {62},
  pages = {3565--3583},
  year = {2020},
  publisher = {Springer}
}

@article{rehman2022efficient,
  title = {Efficient Top-K Identical Frequent Itemsets Mining without Support Threshold Parameter from Transactional Datasets Produced by IoT-Based Smart Shopping Carts},
  author = {Rehman, Saif Ur and Alnazzawi, Noha and Ashraf, Jawad and Iqbal, Javed and Khan, Shafiullah},
  journal = {Sensors},
  volume = {22},
  number = {20},
  pages = {8063},
  year = {2022},
  publisher = {MDPI}
}

@inproceedings{shor1994algorithms,
  author    = {Shor, Peter W.},
  title     = {Algorithms for Quantum Computation: Discrete Logarithms and Factoring},
  booktitle = {Proceedings of the 35th Annual Symposium on Foundations of Computer Science (FOCS)},
  year      = {1994},
  pages     = {124--134},
  publisher = {IEEE},
  doi       = {10.1109/SFCS.1994.365700}
}

@article{singh2019sparkapriori,
  author  = {Singh, Pankaj and Singh, Sudhakar and Mishra, P. K. and Garg, Rakhi},
  title   = {A Data Structure Perspective to the RDD-based Apriori Algorithm on Spark},
  journal = {arXiv preprint arXiv:1908.01338},
  year    = {2019},
  url     = {https://arxiv.org/abs/1908.01338}
}

@inproceedings{uno2004lcm,
  title = {LCM ver. 2: Efficient mining algorithms for frequent/closed/maximal itemsets},
  author = {Uno, Takeaki and Kiyomi, Masashi and Arimura, Hiroki and others},
  booktitle = {Proceedings of the IEEE ICDM Workshop on Frequent Itemset Mining Implementations (FIMI)},
  year = {2004}
}

@article{vo2012dbv,
  title = {DBV-Miner: A Dynamic Bit-Vector approach for fast mining frequent closed itemsets},
  author = {Vo, Bay and Hong, Tzung-Pei and Le, Bac},
  journal = {Expert Systems with Applications},
  volume = {39},
  number = {8},
  pages = {7196--7206},
  year = {2012},
  publisher = {Elsevier}
}

@article{wan2024ptf,
  author  = {Wan, Xiaoyu and Liu, Qing and Li, Guoliang},
  title   = {Efficient Top-$k$ Frequent Itemset Mining on Massive Data},
  journal = {Data Science and Engineering},
  volume  = {9},
  pages   = {177--203},
  year    = {2024},
  doi     = {10.1007/s41019-024-00241-2}
}

@article{wright2019ionq,
  author  = {Kevin Wright and others},
  title   = {Benchmarking an 11-Qubit Quantum Computer},
  journal = {Nature Communications},
  volume  = {10},
  number  = {1},
  pages   = {5464},
  year    = {2019},
  doi     = {10.1038/s41467-019-13534-2}
}

@article{yu2016quantum,
  title = {Quantum algorithm for association rules mining},
  author = {Yu, Chao-Hua and others},
  journal = {Physical Review A},
  volume = {94},
  number = {4},
  pages = {042311},
  year = {2016},
  publisher = {APS}
}

@inproceedings{zaki2002charm,
  title = {CHARM: An efficient algorithm for closed itemset mining},
  author = {Zaki, Mohammed J and Hsiao, Ching-Jui},
  booktitle = {Proceedings of the 2002 SIAM International Conference on Data Mining},
  pages = {457--473},
  year = {2002},
  organization = {SIAM}
}

@article{zhang2021multi,
  title = {Multi-objective optimization for high-dimensional maximal frequent itemset mining},
  author = {Zhang, Yalong and Yu, Wei and Ma, Xuan and Ogura, Hisakazu and Ye, Dongfen},
  journal = {Applied Sciences},
  volume = {11},
  number = {19},
  pages = {8971},
  year = {2021},
  publisher = {MDPI}
}

@article{zhang2021right,
  title = {Right-hand side expanding algorithm for maximal frequent itemset mining},
  author = {Zhang, Yalong and Yu, Wei and Zhu, Qiuqin and Ma, Xuan and Ogura, Hisakazu},
  journal = {Applied Sciences},
  volume = {11},
  number = {21},
  pages = {10399},
  year = {2021},
  publisher = {MDPI}
}

@inproceedings{diaby2017sfpg,
  author    = {Amadou Diaby and Ibou Camara and Mamadou Ndiaye},
  title     = {S-FPG: An Efficient Parallel FP-Growth Algorithm under Apache Spark},
  booktitle = {Proceedings of the 2017 2nd International Conference on Cloud Computing and Big Data Analysis (ICCCBDA)},
  year      = {2017},
  pages     = {360--364},
  publisher = {IEEE},
  doi       = {10.1109/ICCCBDA.2017.7951891}
}

@article{zida2016efim,
  author  = {Soufiane Zida and Philippe Fournier-Viger and Jingwei Lin and Chengwei Wu and Vincent S. Tseng},
  title   = {EFIM: A Highly Efficient Algorithm for High-Utility Itemset Mining},
  journal = {Knowledge and Information Systems},
  volume  = {51},
  number  = {2},
  pages   = {595--625},
  year    = {2017},
  doi     = {10.1007/s10115-016-0986-0}
}

@article{wu2022ubpminer,
  author  = {Peng Wu and Xinzheng Niu and Philippe Fournier-Viger and Cheng Huang and Bing Wang},
  title   = {UBP-Miner: An Efficient Bit Based High Utility Itemset Mining Algorithm},
  journal = {Knowledge-Based Systems},
  volume  = {248},
  pages   = {108865},
  year    = {2022},
  doi     = {10.1016/j.knosys.2022.108865}
}

@inproceedings{leung2005cantree,
  title={CanTree: a tree structure for efficient incremental mining of frequent patterns},
  author={Leung, CK-S and Khan, Quamrul I and Hoque, Tariqul},
  booktitle={Fifth IEEE International Conference on Data Mining (ICDM'05)},
  pages={8--pp},
  year={2005},
  organization={IEEE}
}

@inproceedings{tanbeer2008cp,
  title={CP-tree: a tree structure for single-pass frequent pattern mining},
  author={Tanbeer, Syed Khairuzzaman and Ahmed, Chowdhury Farhan and Jeong, Byeong-Soo and Lee, Young-Koo},
  booktitle={Pacific-Asia Conference on Knowledge Discovery and Data Mining},
  pages={1022--1027},
  year={2008},
  organization={Springer}
}

@article{yu2022experimental,
  author  = {Yu, Chao-Hua},
  title   = {Experimental Implementation of Quantum Algorithm for Association Rules Mining},
  journal = {IEEE Journal on Emerging and Selected Topics in Circuits and Systems},
  volume  = {12},
  number  = {3},
  pages   = {676--684},
  year    = {2022},
  doi     = {10.1109/JETCAS.2022.3201097}
}

@article{fang2024gpu,
  author  = {Fang, Wei and Jiang, Haipeng and Lu, Hengyang and Sun, Jun and Wu, Xiaojun and Lin, Jerry Chun-Wei},
  title   = {GPU-Based Efficient Parallel Heuristic Algorithm for High-Utility Itemset Mining in Large Transaction Datasets},
  journal = {IEEE Transactions on Knowledge and Data Engineering},
  volume  = {36},
  number  = {2},
  pages   = {652--667},
  year    = {2024},
  doi     = {10.1109/TKDE.2023.3290371}
}

@article{boyer1998tight,
  title={Tight bounds on quantum searching},
  author={Boyer, Michel and Brassard, Gilles and H{\o}yer, Peter and Tapp, Alain},
  journal={Fortschritte der Physik: Progress of Physics},
  volume={46},
  number={4-5},
  pages={493--505},
  year={1998},
  publisher={Wiley Online Library}
}

@inproceedings{kosters2003complexity,
  author    = {Walter A. Kosters and Wim Pijls and Viara Popova},
  title     = {Complexity analysis of depth first and {FP}-growth implementations of {APRIORI}},
  booktitle = {Machine Learning and Data Mining in Pattern Recognition (MLDM)},
  series    = {Lecture Notes in Artificial Intelligence},
  volume    = {2734},
  pages     = {284--292},
  publisher = {Springer},
  year      = {2003}
}

@inproceedings{yang2004complexity,
  author    = {Guizhen Yang},
  title     = {The complexity of mining maximal frequent itemsets and maximal frequent patterns},
  booktitle = {Proceedings of the 10th ACM SIGKDD International Conference on Knowledge Discovery and Data Mining},
  pages     = {344--353},
  year      = {2004}
}

@article{georgopoulos2021noise,
  author  = {Konstantinos Georgopoulos and Clive Emary and Paolo Zuliani},
  title   = {Modeling and simulating the noisy behavior of near-term quantum computers},
  journal = {Physical Review A},
  volume  = {104},
  number  = {6},
  pages   = {062432},
  year    = {2021}
}

@misc{online,
      title={Frequent Itemset Mining with Quantum Computing (Full Version)},
      author={Yen-Hsin Hsu and Ya-Wen Teng and De-Nian Yang and Wang-Chien Lee and Philip S. Yu and Ming-Syan Chen},
      year={2026},
      eprint={2606.09209},
      archivePrefix={arXiv},
      primaryClass={cs.DB},
      url={https://arxiv.org/abs/2606.09209}, 
}

@article{bennett1973logical,
  author  = {Bennett, Charles H.},
  title   = {Logical Reversibility of Computation},
  journal = {IBM Journal of Research and Development},
  volume  = {17},
  number  = {6},
  pages   = {525--532},
  year    = {1973}
}

@book{nielsen2010quantum,
  author    = {Nielsen, Michael A. and Chuang, Isaac L.},
  title     = {Quantum Computation and Quantum Information},
  publisher = {Cambridge University Press},
  address   = {Cambridge, UK},
  edition   = {10th Anniversary},
  year      = {2010}
}

@article{fournier2014spmf,
  author  = {Fournier-Viger, Philippe and Gomariz, Antonio and Gueniche, Ted and Soltani, Azadeh and Wu, Cheng-Wei and Tseng, Vincent S.},
  title   = {{SPMF}: A {Java} Open-Source Pattern Mining Library},
  journal = {Journal of Machine Learning Research},
  volume  = {15},
  pages   = {3389--3393},
  year    = {2014},
  note    = {Datasets: \url{https://www.philippe-fournier-viger.com/spmf/index.php?link=datasets.php}}
}

@misc{ibm_quantum_hardware,
  author       = {{IBM Quantum}},
  title        = {{Hardware for useful quantum computing}},
  howpublished = {\url{https://www.ibm.com/quantum/hardware}},
  note         = {Accessed: 2026-06-11}
}

@misc{microsoft_atom_computing_2024,
  author       = {Svore, Krysta},
  title        = {{Microsoft and Atom Computing offer a commercial quantum machine with the largest number of entangled logical qubits on record}},
  howpublished = {\url{https://azure.microsoft.com/en-us/blog/quantum/2024/11/19/microsoft-and-atom-computing-offer-a-commercial-quantum-machine-with-the-largest-number-of-entangled-logical-qubits-on-record/}},
  month        = nov,
  year         = {2024},
  note         = {Accessed: 2026-06-11}
}

@misc{google_willow_quantum_chip_2024,
  author       = {Neven, Hartmut},
  title        = {{Meet Willow, our state-of-the-art quantum chip}},
  howpublished = {\url{https://blog.google/innovation-and-ai/technology/research/google-willow-quantum-chip/}},
  month        = dec,
  year         = {2024},
  note         = {Accessed: 2026-06-11}
}

\clearpage

\ifincludeappendix
\appendices

\section{Pseudocode and Notations}\label{sec:table}\label{sec:codes}
Procedures~\ref{alg:QPR} and~\ref{alg:QIL} detail the pseudocode of Quantum Preprocessing and Representation (QPR)\footnote{QPR uses bounded-error support estimates only for safe singleton pruning. A singleton is pruned by estimation only when its certified interval lies entirely below the support threshold \(\sigma\); if the interval intersects the threshold, its support is resolved exactly by popcounting the materialized transaction-indicator vector before \(L_1\) is finalized; this step is still covered by the \(\mathcal{O}(NM)\) preprocessing term.} and Quantum Itemset Listing (QIL)\footnote{After measurement returns a newly listed frequent itemset \(c=u\cup w\), its transaction-indicator vector \(v_c=v_u\wedge v_w\) is materialized from the parent vectors and retained. The next level can then load these vectors directly instead of rescanning the database.}, respectively.
Table~\ref{tab:notation} summarizes the notations used in this paper.


\section{Comparative Complexity Analysis}\label{sec:comparison}

In this section, we compare the computational complexity and operational characteristics of QFM against prior FIM algorithms. Table~\ref{tab:comparison} summarizes representative methods. 

\textbf{Apriori-like Algorithms.}
Apriori-like methods generate candidate itemsets level by level and then count their occurrences in the database. In the worst case, the number of explored candidates grows exponentially with the maximum itemset length: level $r$ examines up to $\binom{N}{r}$ candidates, and checking one level-$r$ candidate against one transaction costs $\Theta(r)$ item comparisons, giving the strict transaction-dependent verification term $\mathcal{O}(M\sum_{r\le\hat{k}}r\binom{N}{r})$ in addition to database preprocessing, while join-and-prune candidate generation adds \(\mathcal{O}(\sum_{r\le\hat{k}}r^{2}\binom{N}{r})\), following standard Apriori-family complexity accounting~\cite{kosters2003complexity}.
Although quantum counting can reduce the per-candidate dependence on \(M\), transferring this paradigm into quantum computing still leaves the explicit candidate-generation loop outside the quantum procedure and does not provide level-wise state reuse; qARM~\cite{yu2016quantum} exemplifies this hybrid design, as it retains the classical candidate-generation loop. 
Consequently, qARM 
does not remove the classical candidate loop or reuse verified itemset states across levels.

\textbf{FP-growth-like Algorithms.}
FP-growth compresses transactions into an FP-tree and avoids explicit candidate generation in many practical cases~\cite{han2004fptree}. Nevertheless, in the query-count model of~\cite{kosters2003complexity}, each queried pair \((A,j)\), where \(j\) is larger than the largest item in \(A\), corresponds to the unique \((|A|{+}1)\)-itemset \(A\cup\{j\}\). Hence, level \(r\le\hat{k}{-}1\) contributes at most \(\binom{N}{r+1}\) pairs, each requiring at most \(M\) membership queries. After including preprocessing and output-writing costs, the worst-case mining cost is bounded by the strict term \(\mathcal{O}(NM+M\sum_{r=2}^{\hat{k}}\binom{N}{r}+\sum_{r\le\hat{k}}r\binom{N}{r})\) in Table~\ref{tab:comparison}, which follows the worst-case closure of the exact query-count formula of~\cite{kosters2003complexity} together with the cost of writing the output itemsets.\footnote{The two index forms are one substitution apart: parent level \(r\) contributes at most \(\binom{N}{r+1}\) pairs, and summing over \(r\le\hat{k}{-}1\) with \(s=r{+}1\) gives \(\sum_{s=2}^{\hat{k}}\binom{N}{s}\), written with index \(r\) in Table~\ref{tab:comparison}.}
These exponential worst-case terms reflect the output size of the problem itself. When all nonempty itemsets are frequent, the explicit output contains \(2^{N}{-}1\) itemsets and \(\sum_{r\le N}r\binom{N}{r}=N2^{N-1}\) item identifiers, so every exact algorithm that lists all frequent itemsets inherits an exponential output-size lower bound.\footnote{Counting the maximal frequent itemsets is \#P-complete~\cite{yang2004complexity}. Implementations differ in the additional non-output candidates or projections they explore and in their per-pattern processing cost, but none avoids this output-scale term.}
Although these structures explain the practical efficiency of FP-growth, they also make the paradigm difficult to transfer into quantum computing. The workflow relies on tree construction, projected databases, and recursive pointer traversal, whose data-dependent and destructive updates do not naturally support reversible threshold marking for support verification.

\begingroup
\floatname{algorithm}{Procedure}
\begin{algorithm}[t]
  \caption{\rB{Quantum Preprocessing and Representation (QPR)}}
  \label{alg:QPR}
  \begin{algorithmic}[1]
    \Require transaction database $D$, threshold $\sigma$
    \Ensure frequent 1-itemsets $L_1$ 
    \State Construct incidence matrix $X\in\{0,1\}^{N\times M}$ where $X_{i,t}=1$ iff transaction $t$ contains item $i$
    \State $L_1\gets\emptyset$ 
    \For{each item $i$}
        \State Define $\mathcal O^{(i)}_{\sf count}: |t,b\rangle \mapsto |t,b\oplus X_{i,t}\rangle$
        \State Obtain support $s_i$ with $\mathcal O^{(i)}_{\sf count}$
        \If{$s_i \ge \sigma$}
            \State $L_1\gets L_1\cup\{i\}$
        \EndIf
    \EndFor
    \State \Return $L_1$ 
  \end{algorithmic}
\end{algorithm}
\endgroup

\begingroup
\floatname{algorithm}{Procedure}
\begin{algorithm}[t]
  \caption{\rB{Quantum Itemset Listing (QIL)}}
  \label{alg:QIL}
  \begin{algorithmic}[1]
    \Require candidate preparation $U_{\rm cand}^{(k+1)}$, threshold oracle $\mathcal O^{(k+1)}_{\ge\sigma}$
    \Ensure frequent candidates $L_{k+1}$
    \State Determine (or upper-bound) the number of marked candidates using $\mathcal O^{(k+1)}_{\ge\sigma}$
    \State $L_{k+1}\gets\emptyset$, $E\gets\emptyset$ \Comment{$E$ is an exclusion set}
    \While{more marked candidates remain}
        \State Run amplitude amplification using $U_{\rm cand}^{(k+1)}$, $\mathcal O^{(k+1)}_{\ge\sigma}$, and exclusion set $E$
        \State Measure a candidate $c$
        \If{$c\notin E$}
            \State $L_{k+1}\gets L_{k+1}\cup\{c\}$; $E\gets E\cup\{c\}$
        \EndIf
    \EndWhile
    \State \Return $L_{k+1}$
  \end{algorithmic}
\end{algorithm}
\endgroup

\begin{table*}[h]
\centering
\small
\caption{Summary of Notations.}
\label{tab:notation}
\begin{tabularx}{\textwidth}{@{} l X @{}}
\toprule
\textbf{Symbol} & \textbf{Description} \\
\midrule
\multicolumn{2}{@{}l}{\textit{Data and FIM symbols}} \\ \addlinespace[0.5ex]
$D$ & Transaction database, represented as an $N\times M$ binary incidence matrix $X\in\{0,1\}^{N\times M}$. \\
$X_{i,j}$ & Entry of $X$: $1$ if item $i$ appears in transaction $j$, and $0$ otherwise. \\
$N$ & Number of distinct items. \\
$M$ & Number of transactions. \\
$i$ & Index of a single item. \\
$I$ & Candidate itemset, i.e., a nonempty subset of the item universe. \\
$v_I\in\{0,1\}^M$ & Transaction-indicator vector of itemset $I$, where $v_I[j]=1$ iff $I\subseteq T_j$. \\
$\operatorname{sup}(I)$ & Support count of $I$, equal to $\|v_I\|_1$. \\
$\sigma$ & Minimum support threshold. \\
$L_k$ & Set of frequent $k$-itemsets extracted after QIL. \\
$C_k$ & Set of candidate $k$-itemsets considered at level $k$. \\
$P_k$ & Prefix-sharing parent-pair set used by \rB{QSV} to generate candidates at level $k+1$. \\
$\hat{k}$ & \rY{Largest itemset length whose candidate domain is explored; the terminal level may yield no frequent output.} \\
$B_{\hat{k}}$ & \rY{Total number of level-wise candidates explored by QFM} up to length $\hat{k}$, i.e., $\sum_{k=1}^{\hat{k}-1}|C_{k+1}|$. \\
\rG{$Q_{\hat{k}}$} & \rY{Total number of amplitude-amplification rounds over levels}\rG{, i.e., $\sum_{k=1}^{\hat{k}-1}\sqrt{|C_{k+1}|\,\max\{1,|L_{k+1}|\}}$; the output-sensitive refinement of $B_{\hat{k}}$, satisfying $Q_{\hat{k}}\le B_{\hat{k}}$.} \\
\rB{QPR} & \rB{Quantum Preprocessing and Representation}; initialization stage that identifies frequent 1-itemsets. \\
\rB{QSV} & \rB{Quantum Support Verification}; support-verification stage that prepares the superposition over valid candidates and threshold-marking oracle. \\
\rB{QIL} & Quantum Itemset Listing; extraction stage that lists the candidates marked as frequent. \\
\midrule
\multicolumn{2}{@{}l}{\textit{Oracle symbols}} \\ \addlinespace[0.5ex]
$\mathcal{O}_{\sf count}$ & Counting oracle used in QPR for item-support estimation and bit-vector construction. \\
$\mathcal{O}_{\sf load}^{(k)}$ & Transaction load oracle at level $k$. \\
$\mathcal{O}_{\ge\sigma}^{(k)}$ & Threshold-marking oracle at level $k$. \\
$U_{\rm cand}^{(k)}$ &  Superposition preparation unitary at level $k$. \\

\bottomrule
\end{tabularx}
\end{table*}

\textbf{QFM.}
We analyze QFM from the perspectives of time complexity and resource footprint. In terms of time complexity, QFM has cost $\mathcal{O}(NM+Q_{\hat{k}}(\log B_{\hat{k}}+\log M))$, where \(Q_{\hat{k}}\) denotes the total number of amplitude-amplification rounds aggregated over levels and serves as an output-sensitive refinement of the explored candidate count \(B_{\hat{k}}\), with \(Q_{\hat{k}}\le B_{\hat{k}}\). This expression keeps the remaining candidate-growth dependence explicit while showing that transaction-dependent support verification enters through the logarithmic-depth threshold-marking oracle.\footnote{If the inter-level materialization of retained transaction-indicator vectors is charged as sequential bit-level work, it contributes an additional output-sensitive term \(\mathcal{O}(M\sum_k |L_k|)\).}

\begin{table*}[t]
\centering
\scriptsize\renewcommand{\arraystretch}{0.92}
\caption[Complexity comparison of Apriori, qARM, FP-growth, and QFM.]{Comparison of Apriori, \rB{qARM}, FP-growth, and QFM in terms of time complexity, space complexity, and operational characteristics.}
\label{tab:comparison}
\begin{tabularx}{\textwidth}{@{}l>{\raggedright\arraybackslash}X>{\raggedright\arraybackslash}X>{\raggedright\arraybackslash}X>{\raggedright\arraybackslash}X@{}}
\toprule
\textbf{Category} & \textbf{Apriori (Traditional)\rY{~\cite{Agrawal94}}} & \rB{\textbf{qARM~\cite{yu2016quantum,yu2022experimental}}} & \textbf{FP-growth\rY{~\cite{han2000fp,han2004fptree}}} & \textbf{QFM} \\ \midrule
\textbf{Storage} & Original Database & \rB{Database via an assumed black-box access oracle 
} & FP-Tree & Transaction-indicator vectors \\ \midrule
\textbf{Time complexity} &
\rY{\(\mathcal{O}(NM+M\sum_{r\le\hat{k}}r\tbinom{N}{r}\allowbreak+\sum_{r\le\hat{k}}r^{2}\tbinom{N}{r})\)} &
\rY{\(\mathcal{O}(\sqrt{\sigma(M{-}\sigma{+}1)}\allowbreak\sum_{k}(k{+}1)\allowbreak\sqrt{|C_{k+1}|\max\{1,|L_{k+1}|\}})\) oracle queries for \rY{bounded-error evaluation of the exact threshold predicate} (see notes)} &
\rY{\(\mathcal{O}(NM+M\sum_{r=2}^{\hat{k}}\tbinom{N}{r}\allowbreak+\sum_{r\le\hat{k}}r\tbinom{N}{r})\)} &
\(\rR{\mathcal{O}(NM+Q_{\hat{k}}(\log B_{\hat{k}}+\log M))}\) \\
\textbf{Space complexity} &
\rY{\(\mathcal{O}(NM+\hat{k}\max_{r\le\hat{k}}\tbinom{N}{r})\)} &
\rB{Not analyzed at circuit level (implementation reported for $2{\times}2$ and $4{\times}4$ databases~\cite{yu2022experimental})} &
\(\mathcal{O}(NM)\) &
\rY{\(\mathcal{O}(\log |C_{k+1}|{+}M)\)} workspace qubits per oracle call \\ \midrule 
\textbf{Pros} &
Simple implementation &
\rB{Parallel amplitude estimation over all candidates} &
Efficient representation for shared prefixes &
State reuse and shallow support verification \\
\textbf{Cons} &
Inefficient with many candidates &
\rB{Classical candidate generation between levels; approximate thresholding; no reuse across levels} &
Tree construction and recursion overhead &
Requires quantum circuit resources \\ \bottomrule
\end{tabularx}
\vspace{0pt}
\begin{minipage}{0.96\textwidth}\scriptsize\rY{%
Notes: $\hat{k}$ denotes \rY{the largest itemset length whose candidate domain is explored (the terminal level may yield no frequent output)}. For QFM, $B_{\hat{k}}=\sum_{k=1}^{\hat{k}-1}|C_{k+1}|$ is the \rY{total number of explored level-wise candidates} and $Q_{\hat{k}}=\sum_{k=1}^{\hat{k}-1}\sqrt{|C_{k+1}|\max\{1,|L_{k+1}|\}}\le B_{\hat{k}}$ is \rY{the total number of amplification rounds over levels (its output-sensitive refinement)}. The binomial sums are strict $\hat{k}$-parameterized worst-case pattern-growth bounds for the classical baselines\rY{, derived in the text of this appendix (following~\cite{kosters2003complexity}); the FP-growth space entry follows the FP-tree size bound of at most one node per frequent-item occurrence~\cite{han2004fptree}}. Classical entries count RAM operations under the item-comparison containment model, where a level-$r$ candidate contributes a $\Theta(r)$ containment factor; the Apriori cell's third term is join-and-prune candidate generation, dominated by the counting term when $M=\Omega(N)$. \rY{At $\hat{k}=N$, the Apriori entry is $\Theta(NM2^{N})$ when $M=\Omega(N)$, via $\sum_{r\le N}r\binom{N}{r}=N2^{N-1}$ \rY{and $\sum_{r\le N}r^{2}\binom{N}{r}=N(N{+}1)2^{N-2}$}; the FP-growth entry reduces to $\Theta(M2^{N})$ membership queries plus $\Theta(N2^{N-1})$ output writes, i.e., $\mathcal{O}((M{+}N)2^{N})$. The FP-growth time entry counts database membership queries in the model of~\cite{kosters2003complexity} plus output writes. The Apriori space cell reports peak level-wise working storage; the FP-growth cell reports the resident FP-tree only, with the full depth-first recursion chain reaching up to a $\hat{k}$-factor more unless conditional structures are released and rebuilt; retained output is excluded throughout.}

The qARM entry instantiates the published per-level query bound $\mathcal{O}(\sum_k (k{+}1)\sqrt{|C_{k+1}||L_{k+1}|}/\epsilon)$ under the \rY{bounded-error exact-threshold} guarantee used for comparison with QFM. Boundary classification requires estimating support fractions $a=\operatorname{sup}(I)/M$ to additive error $\mathcal{O}(1/M)$; amplitude estimation gives the stated $\sqrt{\sigma(M{-}\sigma{+}1)}$ factor. QFM's verification cost is exact and logarithmic in $M$ by Lemma~\ref{lem:qsv_depth}. 
As published, qARM fixes one precision $\epsilon$ for all candidates: exact classification then forces $1/\epsilon=\Theta(M)$ once any support is near $M/2$, while its default constant-$\epsilon$ setting yields only $\epsilon$-approximate thresholding, a weaker guarantee than the exact decisions compared here.}%
\end{minipage}
\end{table*}



We next derive another form of the output-sensitive bound for QFM by expressing the explored candidate growth through two run-dependent quantities: the maximum level-wise frequent set size and the maximum number of frequent extensions sharing the same prefix. These quantities describe the structure encountered during the actual mining process, rather than additional inputs to QFM.\footnote{Independently of these run-dependent quantities, Theorem~\ref{thm:complexity_fim} implies a fully general time bound. Since \(B_{\hat{k}}\le\sum_{k=2}^{\hat{k}}\binom{N}{k}=\mathcal{O}(2^N)\) and \(Q_{\hat{k}}\le B_{\hat{k}}\), the same accounting gives \(\mathcal{O}(NM+N2^N)\), using \(2^N\log M\le N2^N+M\).}

\begin{corollary}[Bounded candidate growth]
\label{cor:bounded_level}
\rY{For any database and threshold, let $B=\max_k |L_k|$. Under the fixed item order of Section~\ref{subsec:QFc}, the \emph{prefix} of a $k$-itemset is the set of its first $k{-}1$ items; let $d$ be the maximum number of frequent $k$-itemsets sharing the same prefix, over all explored levels. Under the time accounting of Theorem~\ref{thm:complexity_fim}, QFM runs in}
\begin{equation*}
\rY{\mathcal{O}\left(NM+\hat{k}dB(\log(dB)+\log M)\right).}
\end{equation*}
\end{corollary}
\begin{proof}
\rY{For each level $k$, each $k$-itemset has exactly one prefix (its first $k{-}1$ items under the fixed order), so the prefix classes are pairwise disjoint and partition $L_k$; e.g., under the order $1<2<3$, the frequent $2$-itemsets $\{1,2\}$ and $\{1,3\}$ lie in the class of prefix $\{1\}$ and join to the candidate $\{1,2,3\}$, whereas $\{1,3\}$ and $\{2,3\}$ have different prefixes and are never joined. QSV forms candidates only by joining parent itemsets in the same prefix class, and each candidate $(k{+}1)$-itemset arises from exactly one pair in exactly one class, so no pair is counted twice. If a class contains $m_p$ itemsets, it contributes $\binom{m_p}{2}$ parent pairs. Since $m_p\le d$,}
\begin{equation*}
\rY{\binom{m_p}{2}\le \frac{d-1}{2}m_p.}
\end{equation*}
\rY{Summing over prefix classes gives}
\begin{equation*}
\rY{|C_{k+1}|\le \frac{d-1}{2}|L_k|\le \frac{d-1}{2}B.}
\end{equation*}
\rY{Therefore $B_{\hat{k}}=\sum_{k=1}^{\hat{k}-1}|C_{k+1}|\le \frac{\hat{k}(d-1)}{2}B$. Because $Q_{\hat{k}}\le B_{\hat{k}}$, Theorem~\ref{thm:complexity_fim} gives}
\begin{equation*}
\rY{\mathcal{O}\left(NM+\hat{k}dB(\log B_{\hat{k}}+\log M)\right).}
\end{equation*}
\rY{Finally, $\log B_{\hat{k}}=\mathcal{O}(\log(dB))$. If the terminal explored level yields no frequent output, then $L_{\hat{k}-1}\neq\emptyset$ and a frequent $(\hat{k}{-}1)$-itemset has $\hat{k}{-}1$ frequent singleton subsets, so $\hat{k}-1\le|L_1|$; otherwise $L_{\hat{k}}\neq\emptyset$ and $\hat{k}\le|L_1|$ directly. In either case $\hat{k}\le |L_1|+1\le B+1$ and $\log(\hat{k}d)=\mathcal{O}(\log(dB))$. Note that whenever any candidate level is explored, the empty prefix already admits $|L_1|$ frequent extensions, so $d\ge|L_1|\ge2$ and $\log(dB)\ge1$; otherwise the iterative term vanishes and the bound holds trivially. The corollary follows.}
\end{proof}

In terms of resource footprint, we analyze the space complexity of QFM from the active-workspace perspective.\footnote{The classical persistent storage consists mainly of retained transaction-indicator vectors and output and bookkeeping metadata. For the level-wise procedure, QFM retains the transaction-indicator vectors of the current parent level; with level-wise eviction, this requires \(\mathcal{O}(M\max_k |L_k|)\) classical bits. If an implementation retains all levels for later re-mining or audit, the storage becomes \(\mathcal{O}(M\sum_k |L_k|)\) bits. QFM also maintains metadata, such as itemset descriptors for prefix-sharing joins and exclusion indices used by QIL, e.g., \(\mathcal{O}(\max_k k|L_k|)\) item identifiers or \(\mathcal{O}(\max_k |L_k|\log B_{\hat{k}})\) bits when stored as candidate indices. These metadata terms are dominated by the \(M\)-bit transaction-vector storage when \(M\) is at least the descriptor width.}

\rY{\begin{theorem}[Space complexity of QFM]
\label{thm:space_fim}
At level $k{+}1$, one call of the QSV threshold-marking oracle uses
\begin{equation*}
\mathcal{O}(\log |C_{k+1}|+M)
\end{equation*}
active workspace qubits. These registers are uncomputed at the end of the call, so this active footprint is independent of the number of oracle invocations. 
\end{theorem}
\begin{proof}
The active registers in a marking call are the CSI register, the transaction-indicator workspace, the support-counting workspace, and a constant number of flag qubits. The CSI register uses \rY{$\mathcal{O}(\log |C_{k+1}|)$} qubits. The oracle loads the two parent transaction-indicator vectors and computes their bitwise intersection, using \rY{a constant number of $M$-qubit transaction-vector registers ($\mathcal{O}(M)$ qubits in total)}. The final support counter and comparator use \rY{$\mathcal{O}(\log M)$-qubit registers}; any reduction-tree ancillas are $\mathcal{O}(M)$ and are absorbed into the $M$-scale workspace term.
After comparison, the inverse load, intersection, popcount, and comparison blocks uncompute all temporary registers. 
This proves the stated space bound.
\end{proof}}

\begin{figure*}[t]
\centering
{\revisioncolor{yellow!50!black}%
\begin{tikzpicture}
\begin{axis}[
    ymode=log,
    log basis y={10},
    width=0.90\textwidth,
    height=5.3cm,
    ymin=0.05,
    ymax=120000,
    ylabel={Time (s, log scale)},
    xmin=-0.75,
    xmax=11.55, 
    xtick={0,2.6,5.2,7.8,10.4},
    xticklabels={Accidents,Connect-4,Pumsb,Pumsb-star,Chess},
    xticklabel style={font=\scriptsize, rotate=12, anchor=east, yshift=-2mm},
    ymajorgrids=true,
    grid style={dashed,gray!35},
    legend style={at={(0.5,1.12)}, anchor=south, legend columns=6, draw=none, font=\scriptsize},
    tick label style={font=\scriptsize},
    label style={font=\scriptsize},
    axis on top=false,
]
\addlegendimage{area legend, draw=black, fill=black!8}\addlegendentry{QFM}
\addlegendimage{area legend, draw=black, fill=black!28}\addlegendentry{qARM}
\addlegendimage{area legend, draw=black, fill=black!45}\addlegendentry{FP-growth}
\addlegendimage{area legend, draw=black, fill=black!60}\addlegendentry{HAMM}
\addlegendimage{area legend, draw=black, fill=black!75}\addlegendentry{Eclat}
\addlegendimage{area legend, draw=rkR, fill=black!92}\addlegendentry{PHA-HUIM} 

\def\figbase{0.05}
\newcommand{\timebar}[4]{\filldraw[fill=#4, draw=black, line width=0.30pt] (axis cs:#1,\figbase) rectangle (axis cs:#2,#3);}

\draw[densely dashed, gray!38] (axis cs:1.30,\figbase) -- (axis cs:1.30,120000);
\draw[densely dashed, gray!38] (axis cs:3.90,\figbase) -- (axis cs:3.90,120000);
\draw[densely dashed, gray!38] (axis cs:6.50,\figbase) -- (axis cs:6.50,120000);
\draw[densely dashed, gray!38] (axis cs:9.10,\figbase) -- (axis cs:9.10,120000);

\timebar{-0.50}{-0.34}{51.71}{black!8}
\timebar{-0.25}{-0.09}{24162.49}{black!28}
\timebar{0.00}{0.16}{169.12}{black!45}
\timebar{0.25}{0.41}{197.53}{black!60}
\timebar{0.50}{0.66}{30525.81}{black!75}
\filldraw[fill=black!92, draw=rkR, line width=0.5pt] (axis cs:0.75,\figbase) rectangle (axis cs:0.91,147.51); 

\timebar{2.10}{2.26}{11.43}{black!8}
\timebar{2.35}{2.51}{5481.65}{black!28}
\timebar{2.60}{2.76}{45.69}{black!45}
\timebar{2.85}{3.01}{31.60}{black!60}
\timebar{3.10}{3.26}{7509.78}{black!75}
\filldraw[fill=black!92, draw=rkR, line width=0.5pt] (axis cs:3.35,\figbase) rectangle (axis cs:3.51,22.417); 

\timebar{4.70}{4.86}{61.62}{black!8}
\timebar{4.95}{5.11}{37464.10}{black!28}
\timebar{5.20}{5.36}{206.04}{black!45}
\timebar{5.45}{5.61}{202.93}{black!60}
\timebar{5.70}{5.86}{36112.53}{black!75}
\filldraw[fill=black!92, draw=rkR, line width=0.5pt] (axis cs:5.95,\figbase) rectangle (axis cs:6.11,98); 

\timebar{7.30}{7.46}{3.741}{black!8}
\timebar{7.55}{7.71}{622.75}{black!28}
\timebar{7.80}{7.96}{8.238}{black!45}
\timebar{8.05}{8.21}{5.125}{black!60}
\timebar{8.30}{8.46}{471.60}{black!75}
\filldraw[fill=black!92, draw=rkR, line width=0.5pt] (axis cs:8.55,\figbase) rectangle (axis cs:8.71,4.577); 

\timebar{9.90}{10.06}{0.105}{black!8}
\timebar{10.15}{10.31}{38.03}{black!28}
\timebar{10.40}{10.56}{1.329}{black!45}
\timebar{10.65}{10.81}{2.283}{black!60}
\timebar{10.90}{11.06}{6.124}{black!75}
\filldraw[fill=black!92, draw=rkR, line width=0.5pt] (axis cs:11.15,\figbase) rectangle (axis cs:11.31,0.439); 
\end{axis}
\end{tikzpicture}%
}

\caption{\rR{Time comparison on all large datasets (extending Figure~\ref{fig:large_time_rebuttal} with \textit{Pumsb-star}).}}
\label{fig:large_time_full}
\end{figure*}
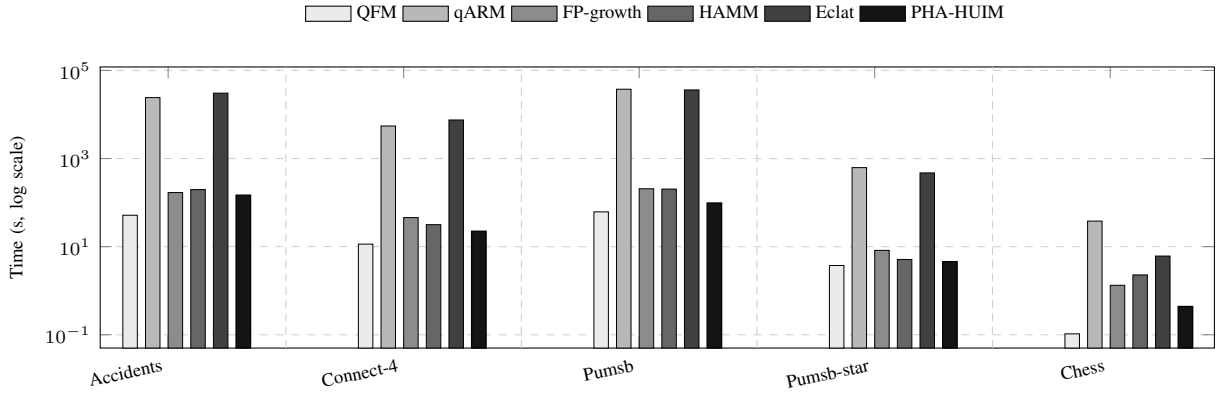

\section{Connection to Classical Mining Paradigms}\label{app:eclat_mapping}
In this section, we discuss how QFM can be extended to support quantum counterparts of classical candidate-based, vertical-data, and tree-based mining algorithms. 

\textbf{Candidate-based mining.}
A direct quantum adaptation of Apriori-style mining would keep the classical level-wise join-and-prune loop and apply quantum search or quantum counting to each generated candidate set. This can reduce the support-counting dependence on the number of transactions, but the candidate-generation loop remains outside the quantum procedure, and the support-counting logic still needs to be embedded inside repeated amplification iterations. Such repeated embedding increases the circuit cost of each level and weakens the benefit obtained from quantum search. \rY{At the other extreme, a flat quantum listing over the full itemset lattice has domain size $2^{N}$ and requires $\widetilde{\mathcal{O}}(\sqrt{2^{N}\max\{1,|L|\}})$ marking-oracle calls; restricting the search to $(k{+}1)$-itemsets reduces the domain to $\binom{N}{k+1}$ but still ignores downward-closure pruning.} In QFM, by contrast, the level-wise candidate structure is instead represented directly by MAC Superposition\rY{, which amplifies only the data-dependent candidate domain $C_{k+1}$ induced by prefix-sharing frequent parents}. The threshold-marking oracle \(\mathcal{O}_{\ge\sigma}\) serves as the reusable marker inside QIL, so support verification is factored into one reversible oracle structure that can be invoked repeatedly during amplification without rebuilding the support-counting logic.

\textbf{Vertical-data mining.}
From the vertical-data perspective used by Eclat, each itemset is represented by a TID-set, and support evaluation is performed by intersecting TID-sets and counting the resulting size. Directly implementing these TID-set intersections as ordinary set operations is not well aligned with reversible circuit execution, since intermediate intersection results must be computed and then cleared without losing reversibility. In contrast, QFM provides a natural quantum-compatible representation through materialized transaction-indicator vectors. In this representation, TID-set intersection corresponds to reversible bit-vector AND, and support counting corresponds to popcount followed by threshold comparison inside \(\mathcal{O}_{\ge\sigma}\). Thus, the Eclat-style support test becomes a reusable reversible oracle over level-wise candidate states, followed by QIL extraction.

\textbf{Tree-based mining.}
FP-growth-style mining obtains much of its practical efficiency from prefix sharing through a compressed FP-tree, projected databases, and recursive traversal. However, the pointer-based traversal and in-place updates used by this workflow do not map naturally to reversible, uncomputation-friendly quantum circuits. By constrast, QFM realizes the corresponding prefix-sharing structure without relying on recursive tree traversal. The prefix-sharing joins is captured by CSI used in MAC Superposition, while the corresponding transaction information is represented by retained transaction-indicator vectors. As a result, QFM replaces the tree-based notion of prefix sharing through level-wise candidate superposition and reversible threshold marking, rather than through FP-tree construction and recursive traversal.


\section{Additional Experimental Results}\label{sec:additional_exp}

\subsection{Time Comparison on \textit{Pumsb-star}}

Figure~\ref{fig:large_time_full} reports the time comparison for all datasets, including \textit{Pumsb-star}. The result on \textit{Pumsb-star} is consistent with the other datasets. QFM remains the fastest method, followed by PHA-HUIM and HAMM, while qARM and Eclat are several orders of magnitude slower.

\subsection{Additional Quantum Resource Usage Comparison}

Table~\ref{tab:yu_quantum_resource_full} extends Table~\ref{tab:yu_quantum_resource} by including the peak logical-qubit footprint. QFM substantially reduces both logical depth and logical gate count compared with qARM, while using the same peak number of logical qubits on both datasets. Thus, the reductions in depth and gate count are not obtained by trading for additional logical-qubit space. Instead, QFM produces circuits that are both shallower and lighter than those of qARM without requiring more logical qubits.

\subsection{Validation on Real Quantum Hardware}
\label{subsec:real_qpu}

To empirically validate the feasibility on physical quantum hardware, we conducted experiments on Amazon Braket. We focus the physical-QPU validation on QFM because this experiment is intended to validate executable quantum circuits under real hardware constraints. qARM cannot be deployed on Amazon Braket since it follows a hybrid workflow with candidate generation performed on classical machines via the level-wise join-and-prune loop. In other words, qARM is not a total quantum solution. 
Due to current hardware constraints,\footnote{For example, IonQ Forte 1 provides \rY{only} up to 36 algorithmic qubits (see \url{https://www.ionq.com/quantum-systems/compare}), which is below the \(39\)--\(43\) peak active logical qubits required by QFM on large datasets, as shown in Table~\ref{tab:yu_quantum_resource_full}.} a small-scale subset was randomly sampled from the \textit{Mushroom}, \textit{Tic-Tac-Toe}, and \textit{Car} datasets for evaluation. 
Table~\ref{tab:braket_qpu_summary} summarizes the physical-QPU validation of QFM on Amazon Braket IonQ Forte 1. Across the QPU runs on the three datasets, the executed workloads use 8--9 logical qubits, reach maximum depths ranging from 410 to 1017, and achieve 100\% correctness in identifying itemsets that satisfy the support threshold; each run was executed with 100 shots, within the device's per-task gate-shot quota.\footnote{The gate-shot quota is a hardware-service limit on how large a submitted task can be, accounting for both circuit size and the number of repeated executions.} These results provide a real quantum hardware validation of QFM under current device constraints.

\begin{table}[t]
\centering
\small
\setlength{\tabcolsep}{2pt}
\caption{\rR{Quantum resource usage comparison (extending Table~\ref{tab:yu_quantum_resource} with the peak \rB{active} logical-qubit footprint, denoted as \rB{Act.\ LQ}).}}
\label{tab:yu_quantum_resource_full}
\begin{tabular}{lrrrrrr}
\toprule
\multirow{2}{*}{Method}
& \multicolumn{3}{c}{\textit{Bike}}
& \multicolumn{3}{c}{\textit{Skin}} \\
\cmidrule(lr){2-4} \cmidrule(lr){5-7}
& \rB{Act.\ LQ} & Depth & \#Gates
& \rB{Act.\ LQ} & Depth & \#Gates \\
\midrule
QFM  & 39 & 1460       & 14760       & 43 & 1820       & 13720 \\
qARM & 39 & $>30000$ & $>30000$ & 43 & $>30000$ & $>30000$ \\
\bottomrule
\end{tabular}
\end{table}

\begin{table}[t]
\centering
\setlength{\tabcolsep}{3pt}
\caption{Physical-QPU validation of QFM on Amazon Braket. Correctness indicates whether QFM reports itemsets that satisfy the support threshold.}
\label{tab:braket_qpu_summary}
\resizebox{\columnwidth}{!}{
\begin{tabular}{lcccc}
\toprule
Backend & QPU runs & Logical qubits & Max depth & Correctness \\
\midrule
IonQ Forte 1 (\textit{Mushroom}) & 11 & 8--9 & 927 & 100\% \\
IonQ Forte 1 (\textit{Tic-Tac-Toe}) & 11 & 8 & 410 & 100\% \\
IonQ Forte 1 (\textit{Car}) & 11 & 8--9 & 1017 & 100\% \\
\bottomrule
\end{tabular}}
\end{table}

To complement the correctness validation, we also examine the physical-QPU execution time across these sampled QPU runs under controlled changes to the sampled data volume and the minimum support threshold. The reported percentages are relative to the original sampled configuration with \(\sigma=2\). The observed trend is consistent with the simulator-based results. For \textit{Mushroom}, \textit{Tic-Tac-Toe}, and \textit{Car}, reducing the sampled data volume by one third sharply lowers the execution time to below \(1\%\) of the original configuration, whereas increasing it by one third leads to a moderate increase to about \(118\%\), reflecting the larger circuit depth required to encode and verify a larger transaction set. The minimum-support threshold shows the opposite trend. Relaxing the threshold to \(\sigma=1\) increases the execution time to about \(107\%\), since more candidates satisfy the support condition and require extraction. Tightening the threshold to \(\sigma=3\) prunes the search space more aggressively, reducing the execution time to about \(29\%\). Overall, the hardware measurements follow similar qualitative behavior observed in simulation: the execution time increases with the sampled data volume, while higher support thresholds reduce the number of surviving candidates and lower the overall execution time.

\subsection{Error Measurement under Simulated Hardware Noise}
\label{subsec:qsv_error}

The 100\% correctness observed on IonQ Forte 1 is consistent with its trapped-ion hardware characteristics, including high gate fidelity and all-to-all connectivity, which help reduce accumulated gate error for the tested circuits.\footnote{Detailed information about IonQ Forte 1 can be found at \url{https://www.ionq.com/quantum-systems/forte}.} Therefore, in addition to the physical-QPU validation, following~\cite{georgopoulos2021noise}, we further evaluate how much hardware noise QFM can tolerate by running a broader set of tested instances on IBM Qiskit Aer with a controlled, increasing amount of hardware error. Each noise level is chosen to represent a class of devices, ranging from low-noise trapped-ion machines (e.g., IonQ) to noisier superconducting ones (e.g., Rigetti).\footnote{On current quantum hardware the two-qubit entangling gates dominate the error budget, so the noise level we vary is the two-qubit gate error rate; the noise model combines a depolarizing part (gate error) with a readout part (measurement error).} For each noise level, we report the correctness of the frequent/infrequent decisions. 

\begin{table}[t]\centering\footnotesize\setlength{\tabcolsep}{2pt}
\caption{Error measurement of QFM under Qiskit Aer noise models with increasing two-qubit gate error (depolarizing \(+\) readout). The noise levels are representative settings for sensitivity analysis rather than measured calibration values of specific hardware devices.}

\label{tab:qsv_error}
\begin{tabular}{lcc}
\toprule
Setting & 2-qubit err & Correct \\ 
\midrule
Aer, ideal                  & \(0\)     & $100\%$ \\ 
Aer, trapped-ion-class (e.g., IonQ Forte)      & \(0.006\) & $99.7\%$ \\ 
Aer, generic                & \(0.010\) & $99.0\%$ \\ 
Aer, stress & \(0.015\) & $92.7\%$ \\ 
\bottomrule
\end{tabular} 
\end{table}

Table~\ref{tab:qsv_error} summarizes the correctness 
under the simulated noise settings. QFM achieves near-100\% correctness in the noiseless, trapped-ion-class, and generic settings, consistent with the IonQ Forte~1 hardware validation. As the two-qubit error rate increases, 
the decision quality degrades accordingly. The correctness remains at or above \(99.0\%\) up to the generic setting, and QFM still retains \(92.7\%\) correctness under the conservative stress setting. Overall, QFM remains reliable in the low-two-qubit-error regime matching the physical-QPU platform used in our validation, while its decision quality degrades gradually as the two-qubit gate error increases.\footnote{For the QIL-stage candidate-verification and readout errors considered here, QFM handles them at the QIL extraction stage through verification and repeated extraction. Each returned itemset is checked against the retained parent transaction-indicator vectors before insertion into \(L_{k+1}\). Thus, a noisy marking or readout event that produces an invalid, padded, or infrequent itemset is rejected rather than reported, and only verified positives update the exclusion set. Missed frequent itemsets can be mitigated by rerunning QIL on the remaining search space with a larger shot or extraction budget.}

\fi

\end{document}